\documentclass[11pt]{article}

\usepackage[T1]{fontenc}
\usepackage{latexsym,color,amsmath,amssymb,amsthm}
\usepackage{floatflt}
\usepackage{color}
\usepackage{times}
\usepackage{euscript}
\usepackage{epsfig}
\usepackage{graphics}
\usepackage{graphicx}
\usepackage{mathpazo}
\usepackage[top=1.25in,bottom=1.25in,left=1.25in,right=1.25in]{geometry}
\usepackage{microtype}
\usepackage{subcaption}

\usepackage{fullpage}
\usepackage{xspace}
\usepackage[shortlabels, inline]{enumitem}

\usepackage{ifpdf}

\graphicspath{{fig/}}

\DeclareGraphicsExtensions{.pdf}

\ifpdf
\else
   \DeclareGraphicsExtensions{.eps}
\fi

\usepackage[colorinlistoftodos,textsize=tiny,textwidth=2cm,color=red!25!white,obeyFinal]{todonotes}
\usepackage{xspace}
\usepackage{soul} % for highlighting, if command has weird errors, register it as below
	\soulregister\cite7
	\soulregister\ref7
\theoremstyle{plain}
\newtheorem{theorem}{Theorem}
\newtheorem{lemma}{Lemma}
\newtheorem{corollary}{Corollary}
\newtheorem{observation}{Observation}

\makeatletter
\long\def\@makecaption#1#2{
    \vskip 10pt
    \setbox\@tempboxa\hbox{{\footnotesize {\bf #1.} #2}}
    \ifdim \wd\@tempboxa >\hsize         % IF longer than one line:
        {\footnotesize {\bf #1.} #2\par}% THEN set as ordinary paragraph.
      \else                              %   ELSE  center.
        \hbox to\hsize{\hfil\box\@tempboxa\hfil}
    \fi}
\makeatother

\title{Faster Algorithms for the Geometric Transportation Problem%
\date{}
\author{
	Pankaj Agarwal\\ 
	Duke University\\
	\small{\texttt{pankaj@cs.duke.edu}}
	\and
	Kyle Fox\\
	The University of Texas at Dallas\\
	\small{\texttt{kyle.fox@utdallas.edu}}
	\and
	Debmalya Panigrahi\\
	Duke University\\
	\small{\texttt{debmalya@cs.duke.edu}}
	\and
	Kasturi Varadarajan\\
	University of Iowa\\
	\small{\texttt{kasturi-varadarajan@uiowa.edu}}
	\and
	Allen Xiao\\
	Duke University\\
	\small{\texttt{axiao@cs.duke.edu}}
}
\footnote{Work by Agarwal, Fox, and Xiao is supported by NSF under grants
  	CCF-15-13816, CCF-15-46392, and IIS-14-08846, by ARO grant W911NF-15-1-0408, 
	and by grant 2012/229 from the U.S.-Israel Binational Science Foundation. 
	Work by Fox and Panigrahi is supported in part by NSF grants CCF-1527084 and CCF-1535972. 
	Work by Varadarajan is supported by NSF awards CCF-1318996 and CCF-1615845}
}

\usepackage{notation}

\begin{document}

\maketitle

\begin{abstract}
Let $\reds,\blues \subset \reals^d$, for constant $d$, 
be two point sets with $|\reds|+ |\blues| = n$, 
and let $\tsupply : \reds \cup \blues \to \nats$ such that 
$\sum_{r \in \reds } \tsupply(r) = \sum_{b \in \blues} \tsupply(b)$
be demand functions over $\reds$ and $\blues$. 
Let $\metric(\cdot,\cdot)$ be a suitable 
distance function such as the $L_p$ distance. 
The transportation problem asks to find a map 
$\transp : \reds \times \blues \to \nats$ such that 
$\sum_{b \in \blues}\transp(r,b) = \tsupply(r)$, 
$\sum_{r \in \reds}\transp(r,b) = \tsupply(b)$, and
$\sum_{r \in \reds, b \in \blues} \transp(r,b) \metric(r,b)$ is minimized. 
We present three new results for the transportation problem when 
$\metric(\cdot,\cdot)$ is any $L_p$ metric:

\begin{itemize}
\item For any constant $\eps > 0$, an $O(n^{1+\eps})$ expected time 
	randomized algorithm that returns a transportation map with expected cost 
	$O(\log^2(1/\eps))$ times the optimal cost.
\item For any $\eps > 0$, a $(1+\eps)$-approximation in 
	$O(n^{3/2}\eps^{-d} \polylog(U) \polylog(n))$ time, 
	where $U = \max_{p\in R\cup B} \tsupply(p)$.
\item An exact strongly polynomial $O(n^2 \polylog n)$ time algorithm, for $d = 2$.
\end{itemize}

\end{abstract}

\section{Introduction}
\label{section:intro}
Let $\reds$ and $\blues$ be two point sets in $\reals^d$ with 
$|\reds| + |\blues| = n$, where $d$ is a constant,
and let $\tsupply: \reds \cup \blues \to \nats$ be a function satisfying
$\sum_{r \in \reds} \tsupply(r) = \sum_{b \in \blues} \tsupply(b)$.
We denote $U := \max_{p \in \reds \cup \blues} \tsupply(p)$.
We call a function $\transp: \reds \times \blues \to \nats$,
a \emph{transportation map} between $\reds$ and $\blues$ if 
$\sum_{b \in \blues}\transp(r, b) = \tsupply(r)$ for all $r \in \reds$
and $\sum_{r \in \reds}\transp(r, b) = \tsupply(b)$ for all $b \in \blues$.
Informally, for a point $r\in \reds$, the value of $\tsupply(r)$
represents the {\em supply} at $r$, while for a point $b\in \blues$,
the value of $\tsupply(b)$ represents the \emph{demand} at $b$. A transportation
map represents a plan for moving the supplies at points in 
$\reds$ to meet the demands at points in $\blues$.

The cost of a transportation map $\transp$ is defined as
$\tfcost(\transp) = \sum_{(r, b) \in \reds \times \blues} \transp(r, b) \metric(r, b)$,
where $\metric(\cdot, \cdot)$ is a suitable distance function 
such as the $L_p$ distance.
The \emph{Hitchcock-Koopmans transportation problem} 
(or simply \emph{transportation problem}) 
on $\inst = (\reds, \blues, \tsupply)$ is to find the 
minimum-cost transportation map for $\inst$, 
denoted $\transp^* := \transp^*(\inst)$.
The optimal cost $\tfcost(\transp^*)$ is often referred to as the 
\emph{transportation distance} or {\em earth mover's distance}.

The transportation problem is a discrete version of 
the so-called \emph{optimal transport}, or \emph{Monge-Kantarovich}, problem,
originally proposed by the French mathematician Gaspard Monge in 1781.
This latter problem has been extensively studied in mathematics since the 
early 20th century. See the book by Villani~\cite{villani2008optimal}.
In addition to this connection, 
the (discrete) transportation problem has a wide range of applications, including 
similarity computation between a pair of images, shapes, and distributions, 
computing the barycenter of a family of distributions, 
finding common structures in a set of shapes, fluid mechanics, 
and partial differential equations. Motivated by these applications,
this problem has been studied extensively in many fields 
including computer vision, computer graphics, machine learning, 
optimization, and mathematics.
See e.g. \cite{
rubner1998metric,
grauman2004fast,
DBLP:conf/icml/CuturiD14,
DBLP:journals/tog/SolomonRGB14,
gramfort2015fast}
and references therein for a few examples.

The transportation problem can be formulated as an
instance of the uncapacitated minimum cost flow problem 
in a complete bipartite graph $\reds \times \blues$ with uncapacitated edges.
The minimum cost flow problem has been widely studied; 
see \cite{LS-flow} for a detailed review of known results.
The uncapacitated minimum cost flow problem in a graph with $n$ vertices and $m$ 
edges can be solved in $O((m + n\log n) n\log n)$ time using 
Orlin's algorithm \cite{orlin1993faster}
or $\tilde{O}(m \sqrt{n} \polylog U)$ 
time\footnote{We use $\tilde{O}(f(n))$ to denote $O(f(n) \polylog(n))$.}
using the algorithm by Lee and Sidford~\cite{LS-flow}.

For transportation in geometric settings,
Atkinson and Vaidya~\cite{atkinson1995using} adapted the Edmonds-Karp 
algorithm to exploit geometric properties, and 
obtained an $\tilde{O}(n^{2.5} \log U)$ time algorithm for any $L_p$-metric,
and $\tilde{O}(n^2)$ for $L_1$, $L_\infty$-metrics.
The Atkinson-Vaidya algorithm was improved using faster
data structures for dynamic nearest-neighbor searching,
first in \cite{agarwal2000vertical} and most recently in
\cite{kaplan2017dynamic},
for a running time of $\tilde{O}(n^2 \log U)$.
Sharathkumar and Agarwal~\cite{sharathkumar2012algorithms}
designed a $(1 + \eps)$-approximation algorithm with a
$\tilde{O}((n \sqrt{nU} + U \log U) \log (n/\eps))$ running time.

More efficient algorithms are known for estimating the 
the optimal cost (earth mover's distance) without computing
the map itself, provided that $U = n^{O(1)}$.
Indyk~\cite{indyk2007near} gave an algorithm to find an $O(1)$-approximate estimate 
in $\tilde{O}(n)$ time with probability at least $1/2$.
Cabello~\etal~\cite{cabello2008matching} reduced the problem to minimum cost flow
on a geometric spanner, obtaining a $(1+\eps)$-approximate estimate in 
$\tilde{O}(n^2)$ time.
Andoni~\etal~\cite{andoni2014parallel} gave a streaming algorithm 
that finds a $(1+\eps)$-approximate estimate in $O(n^{1+o_\eps(1)})$ time.
However, in many applications, one is interested in computing the map itself 
and not just the transportation distance 
\cite{gramfort2015fast, DBLP:conf/icml/CuturiD14}. This is the problem
that we address in this paper.

The special case of the transportation problem where every point has unit 
demand/supply is a minimum-cost bipartite matching problem (assignment problem).
After sequence of papers in the geometric setting,
\cite{vaidya1989geometry, varadarajan1999approximation, 
sharathkumar2012algorithms, agarwal2004near},
a near-linear $\tilde{O}(n)$ time $(1+\eps)$-approximation
was found by Agarwal and Sharathkumar~\cite{sharathkumar2012near}
for geometric bipartite matching. 
On the other hand, before our work, no constant-factor 
approximation in subquadratic time was known for the transportation problem
with arbitrary demands and supplies, even for the special case of $U = O(n^2)$. 

\mparagraph{Our results} 
We present three new results for the geometric transportation problem,
for any $L_p$-metric.

Our first result (Section~\ref{section:grids}) is a randomized algorithm that
for any $\eps > 0$, computes in $O(n^{1+\eps})$ expected time a 
transportation map whose expected cost is 
$O(\log^2 (1/\eps)) \tfcost(\transp^*)$.
The expected cost improves to $O(\log (1/\eps)) \tfcost(\transp^*)$ if the 
spread of $\reds \cup \blues$ is $n^{O(1)}$,
where the \emph{spread} $\Phi$ is the ratio of the 
maximum and the minimum distance between a pair of points.
The overall structure of our algorithm is a simpler version of the 
geometric bipartite matching algorithm by Agarwal and Varadarajan~\cite{agarwal2004near},
but several new ideas are needed to handle arbitrary supplies and demands.

This algorithm can be extended to spaces with 
bounded doubling dimension when the spread is polynomially bounded:
if $\reds, \blues$ lie in a subspace of $\reals^d$ with constant doubling dimension and $\Phi = n^{O(1)}$,
then a modified version of the algorithm finds, 
in $O(n^{1 + \eps})$ expected time, 
a transportation map whose expected cost is $O(\log (1/\eps)) \tfcost(\transp^*)$.
%In addition, we can use ideas from Indyk's algorithm~\cite{indyk2007near} to estimate
%$\tfcost(\transp^*)$ within an $O(1)$ factor in $\tilde{O}(n)$ time when $U$ is arbitrary;
%Indyk's original algorithm estimated cost only for $U = n^{O(1)}$.

Our second result (Section~\ref{section:wspd}) is a
$(1 + \eps)$-approximate, $\tilde{O}(n^{3/2}\eps^{-d} \polylog(U))$ time
algorthim, by reduction to minimum cost flow.
Using a quad-tree based well-separated pair decomposition (WSPD)~\cite{callahan1995decomposition}
of a point set, we construct a graph $\sparseG$ with $O(n)$ vertices
and $O(n/\eps^d)$ edges, and reduce the problem of computing
a $(1 + \eps)$-approximate transportation map to computing
the minimum cost flow in $\sparseG$.
We find a minimum cost flow $\flow^*$ using the Lee-Sidford~\cite{LS-flow} algorithm,
and finally recover a transportation map 
whose cost is close to $\tfcost(\flow^*)$ (in turn, close to $\tfcost(\transp^*)$).
This algorithm extends to spaces with bounded doubling dimension
using the appropriate speedy WSPD construction for such spaces \cite{har2005fast}.
Specifically, if the doubling dimension is $D$ and the spread of 
$\reds \cup \blues$ is $n^{O(1)}$, then a $(1 + \eps)$-approximate
transportation map can be computed in time
$\tilde{O}(n^{3/2}\eps^{-O(D)} \polylog(U))$.
%With a small change in parameters, the same algorithm works for squared Euclidean costs.

Our third result (Section~\ref{section:orlin}) is an exact,
$\tilde{O}(n^2)$ time algorithm for $d=2$,
thereby matching (up to poly-logarithmic factors) the best exact algorithm
for geometric bipartite matching \cite{kaplan2017dynamic}.
This is an implementation of Orlin's strongly polynomial 
minimum cost flow algorithm~\cite{orlin1993faster},
an augmenting-paths algorithm with edge contractions.
A naive application of Orlin's algorithm has a running time of $\tilde{O}(n^3)$,
but by exploiting the geometry of the underlying graph, 
we improve this to $\tilde{O}(n^2)$ time in the plane.

\section{A Near-Linear Approximation}
\label{section:grids}
Let $\inst = (\reds, \blues, \tsupply)$ be an instance of the 
transportation problem in $\reals^d$.
We say that $\inst$ has \emph{bounded spread} if the spread of 
$\reds \cup \blues$ is bounded above by $n^{a}$, for some constant $a > 0$.
We present a randomized recursive algorithm that, given $\inst$ and a parameter 
$\eps > 0$, returns a transportation map in $O(n^{1+\eps})$ expected time whose
expected cost is $O(\log (1/\eps)) \tfcost(\transp^*)$ if $\inst$ has 
bounded spread, and $O(\log^2 (1/\eps)) \tfcost(\transp^*)$ otherwise
(recall that $\transp^*$ is the optimal map).
We assume that $n$ is sufficiently large so that $n^{\eps}$ is 
at least a suitably large constant.

We first give a high-level description of the algorithm without describing 
how each step is implemented efficiently.
Next, we analyze the cost of the transportation map computed by the algorithm.
We then discuss an efficient implementation of the algorithm.
For simplicity, we describe the algorithm and its analysis for dimension 
$d = 2$; the algorithm extends to $d > 2$ in a straightforward manner.

We need the notion of \emph{randomly shifted grids},
as in \cite{arora1996polynomial, agarwal2004near}.
Formally, let $\qtsquare = [a - \ell, a] \times [b - \ell, b]$ 
be a square of side length $\ell$ with $(a, b)$ as its top right corner.
For a parameter $\Delta > 0$ (grid cell sidelength),
set $l = \lceil \log_2 \left(1 + \frac{\ell}{\Delta}\right) \rceil$,
and $L = 2^{l+1} \Delta$.
Let $\qtsquare_L = [a - L, a] \times [b - L, b]$ be the square of
side length $L$ with $(a, b)$ as its top-right corner.
We choose uniformly at random a point $\xi \in [0, \Delta)^2$ and set
$\qtsquare_{shifted} := \qtsquare_L + \xi$.
Note that $\qtsquare \subseteq \qtsquare_{shifted}$.
Let $\rsg(\qtsquare, \Delta)$ be the partition of $\qtsquare_{shifted}$ 
into the uniform grid of side length $\Delta$;
$\rsg(\qtsquare, \Delta)$ has $2^{l+1} \times 2^{l+1}$ grid cells.
$\rsg(\qtsquare, \Delta)$ is called the randomly shifted grid on $\qtsquare$.

\subsection{A high-level description}
\label{subsection:grid_high_level}

A recursive subproblem $\rinst = \rtinst{}{}$
consists of point sets $\rreds$ and $\rblues$,
and a demand function $\rsupply: \rreds \cup \rblues \to \nats$
such that $\rsupply(\rreds) = \rsupply(\rblues)$.
We denote $|\rreds \cup \rblues| = m$.
If $m \leq n^{\eps/4}$, we call $\rinst$ a \emph{base} 
subproblem and compute an optimal transportation using Orlin's algorithm.
Thus, assume that $m > n^{\eps/4}$. 

\begin{figure}
\centering
\includegraphics[width=0.3\linewidth,page=6]{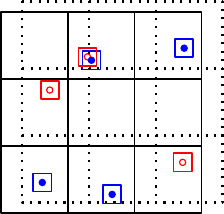}
\caption{Moats, a safe grid (solid), and an unsafe grid (dotted).}
\label{figure:grid_moats}
\end{figure}

Let $\delta = 1/(3d) = 1/6$, a constant. 
We use $\Box$ to denote the smallest axis-aligned square containing $\rreds \cup \rblues$,
and its sidelength $\ell$, and set $\Delta = \ell/m^\delta$. 
\begin{observation}
	A randomly shifted grid $\rsg = \rsg(\Box, \Delta)$ has the following properties.
	\begin{enumerate}
	\item The number of cells is $O([n^\delta]^d) = O(n^{1/3})$.
	\item The diameter of each $\rsg$ cell is $\sqrt{d} \cdot \Delta = O(\Delta)$.
	\item For any two points $r \in \rreds$ and $b \in \rblues$,
		the probability that they lie in different cells of $\rsg$ is $O(\dist{r}{b}/\Delta)$.
	\end{enumerate}
\end{observation}
The first step of the algorithm is to choose a randomly shifted grid 
$\rsg = \rsg(\Box, \Delta)$ that has the following additional property:
any two points in $\rreds \cup \rblues$ that are within a distance of 
$\ell/m^3$ lie in the same grid cell.
We call a grid satisfying this property \emph{safe}.
Algorithmically, we place an axis-parallel square of 
side length $2\ell/m^3$ around every $p \in \rreds \cup \rblues$,
called the \emph{moat} of $p$;
$\rsg$ is safe if none of its grid lines cross any moat 
(see Figure~\ref{figure:grid_moats}).
\begin{observation}
	A randomly shifted grid $\rsg$ is safe with probability $\Omega(1 - 1/m^2)$.
\end{observation}
\begin{proof}
	We can interpret the random shift on $\rsg$ as a composition of 1-dimensional shifts in each dimension.
	These 1-dimensional shifts are chosen uniformly at random from $[0, \ell]$.
	In any single dimension $e_i$, the union of moats has length at most $m \cdot 2\ell/m^3 = 2\ell/m^2$.
	Dividing, the probability that the shift in $e_i$ crossed a grid line with a moat is $O(1/m^2)$.
	There are a constant number of dimensions, so the total failure probability is $O(1/m^2)$.
\end{proof}
If $\rsg$ is not safe, we re-sample random shift until we find one that is safe.

Let $\Pi \subseteq \rsg$ be the set of nonempty grid cells,
i.e., ones that contain at least one point of $\rreds \cup \rblues$.
For each cell $\pi \in \Pi$, we create a recursive instance $\rinst_\pi$,
which we refer to as an \emph{internal subproblem}.
Each $\rinst_\pi$ aims to transport as much as possible within $\pi$.
Whatever we are unable to transport locally within cells of $\Pi$,
we transport globally with a single \emph{external subproblem} $\rinst_\Box$.
We now describe these subproblems in more detail.

For each cell $\pi \in \Pi$, we define the \emph{extra supply} of $\pi$ 
to be the absolute difference between the supply and demand in $\pi$,
denoted by $\excess_\pi$.
Without loss of generality, assume 
$\rsupply(\rreds \cap \pi) \geq \rsupply(\rblues \cap \pi)$,
i.e. that the extra supply of $\pi$ is red.
We will use the entirety of $\rblues \cap \pi$ for the internal 
subproblem, and arbitrarily partition $\rreds \cap \pi$ such that 
$\rsupply(\rblues \cap \pi)$ supply is used for the internal subproblem,
and the remainder (of total supply $\excess_\pi$) 
is used for the external subproblem.
We pick an arbitrary maximal subset of points 
$(\rreds_{ex})_\pi \subseteq \rreds \cap \pi$ such that 
$\rsupply((\rreds_{ex})_\pi) \leq \excess_\pi$, as follows.
Let $\rreds_\pi = (\rreds \cap \pi) \setminus (\rreds_{ex})_\pi$, 
$\rblues_\pi = \rblues \cap \pi$, and $(\rblues_{ex})_\pi = \emptyset$.
If $\rsupply((\rreds_{ex})_\pi) < \excess_\pi$, 
we arbitrarily pick a point $p$ in $\rreds_\pi$, 
and split $p$ into two copies, say $p'$ and $p''$,
with $\rsupply(p') = \excess_\pi - \rsupply((\rreds_{ex})_\pi)$
and $\rsupply(p'') = \rsupply(p) - \rsupply(p')$.
We then add $p'$ to $(\rreds_{ex})_\pi$ and replace $p$ with $p''$ in $\rreds_\pi$.
This step ensures that $\rsupply((\rreds_{ex})_\pi) = \excess_\pi$.
Let $\rsupply_\pi$ be the restriction of $\rsupply$ to 
$\rreds_\pi \cup \rblues_\pi$; 
by construction, $\rsupply_\pi(\rreds_\pi) = \rsupply_\pi(\rblues_\pi)$.
The internal subproblem for $\pi$ is $\rinst_\pi = \rtinst{}{_\pi}$.

We now describe the external subproblem.
Let $\rreds_{ex} = \bigcup_{\pi \in \Pi} (\rreds_{ex})_\pi$, 
$\rblues_{ex} = \bigcup_{\pi \in \Pi} (\rblues_{ex})_\pi$,
and set $\rsupply_{ex}$ as the restriction of $\rsupply$ to $\rreds_{ex}\cup \rblues_{ex}$.
We merge the extra supply in each cell into a single artificial point at the center of the cell.
The resulting transportation instance has relatively few ($O(m^{2\delta})$) points 
and distorts the ``real'' distances by an amount proportional to the side length of the cell.
If $\rsupply(\rreds \cap \pi) > \rsupply(\rblues \cap \pi)$, 
we create a red point $r_\pi$ at the center of $\pi$ 
and define the supply of $r_\pi$, denoted $\rsupply_\Box(r_\pi)$, 
to be $\excess_\pi$.
Similarly, if $\rsupply(\rblues \cap \pi) > \rsupply(\rreds \cap \pi)$, 
we create a blue point $b_\pi$ at the center of $\pi$ with 
$\rsupply_\Box(b_\pi) = \excess_\pi$.
Let $\rreds_\Box$ (resp., $\rblues_\Box$) be the set of red (resp., blue) 
points that were created at the centers of cells in $\Pi$.
We create the external subproblem $\rinst_\Box = \rtinst{}{_\Box}$;
$\rinst_\Box$ acts as an approximate view of the 
actual extra supply instance $\rinst_{ex} = \rtinst{}{_{ex}}$.
See Figure~\ref{figure:grid_subproblems}.

\begin{figure}
\centering
\subcaptionbox{\centering}{\includegraphics[width=0.25\linewidth,page=1]{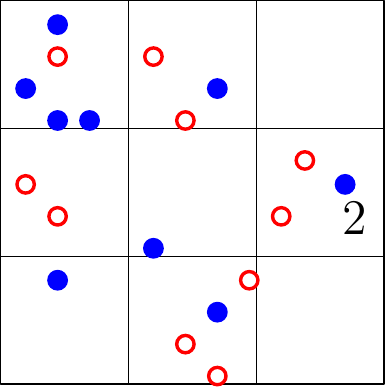}}
\quad
\subcaptionbox{\centering}{\includegraphics[width=0.25\linewidth,page=2]{fig/transp_alg.pdf}}
\quad
\subcaptionbox{\centering}{\includegraphics[width=0.25\linewidth,page=3]{fig/transp_alg.pdf}}
\caption{A subproblem (a) with its internal subproblems (b) and external subproblem (c).}
\label{figure:grid_subproblems}
\end{figure}

For each cell $\pi \in \Pi$, we recursively compute a transportation map 
$\transp_\pi$ on the internal subproblem $\rinst_\pi$.
If the root instance -- the original input to our transportation 
problem -- has bounded spread, we compute an optimal solution 
$\transp_\Box$ for the external subproblem $\rinst_\Box$ using Orlin's algorithm.
If the root instance does not have bounded spread, then we recursively compute 
an approximately optimal solution $\transp_\Box$ for the external subproblem 
$\rinst_\Box$.
Note that irrespective of the spread of the original instance, 
every external subproblem $\rinst_\Box$ has spread bounded by $O(n^{\delta})$,
i.e., has bounded spread.

We categorize subproblems by the number of external subproblems 
in the recursive chain leading to them:
$\rinst$ is \emph{primary} if there are none;
\emph{secondary} if there is exactly one; and \emph{tertiary} if there are two.
All tertiary problems are solved exactly using Orlin's algorithm, as are
base subproblems in the primary and secondary recursion.
See Figure~\ref{figure:grid_primary_secondary}
for a visualization of the recursion tree of the algorithm.

\begin{figure}
\centering
\includegraphics[width=0.7\linewidth]{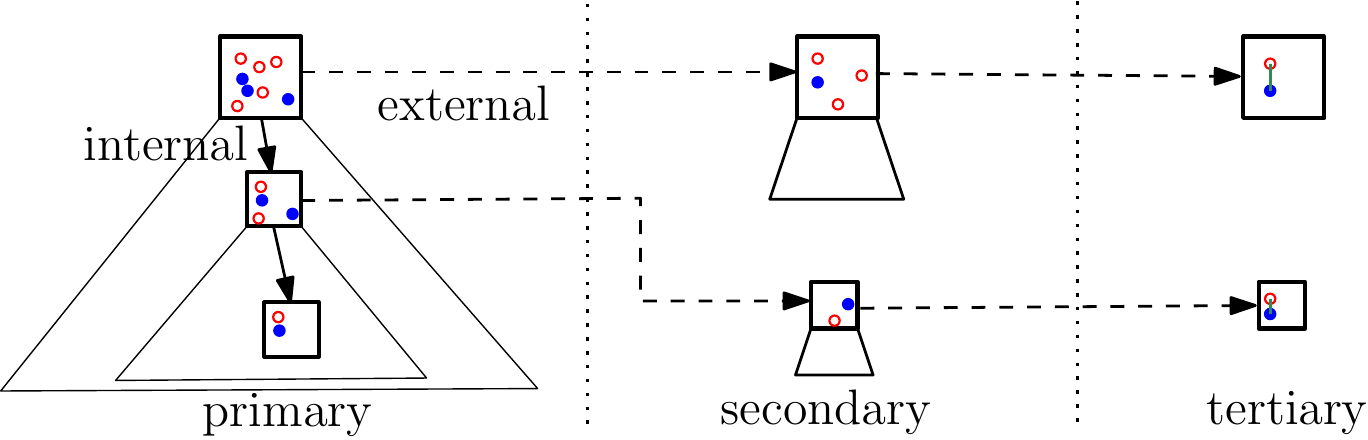}
\caption{Primary-secondary classification of recursive problems.}
\label{figure:grid_primary_secondary}
\end{figure}

Finally, we construct a transportation map $\transp$ for $\rinst$ 
by combining the solutions to the internal and external subproblems.
For a pair $(r, b) \in \rreds_\pi \times \rblues_\pi$, 
we simply set $\transp(r, b) = \transp_\pi(r, b)$.
For the external subproblem, we first convert the transportation map
$\transp_\Box$ on $\rinst_\Box$ into a map for $\rinst_{ex}$, as follows:
For each red point $r_\pi \in \rreds_\Box$ 
(resp., blue point in $\rblues_\Box$), at the center of a cell $\pi \in \Pi$,
we ``redistribute'' the transport from $r_\pi$ (resp., $b_\pi$) to 
the points of $(\rreds_{ex})_\pi$ (resp., $(\rblues_{ex})_\pi$) 
to compute a transportation map $\transp_{ex}$ of $\rinst_{ex}$.
That is, for any $r_\pi, b_\pi \in \rreds_\Box \times \rblues_\Box$, 
we assign the units of $\transp_\Box(r_\pi, b_\pi)$ 
among the pairs in $(\rreds_{ex})_\pi \times (\rblues_{ex})_\pi$ in an arbitrary manner, 
while respecting the demands. We then set $\transp(r, b) = \transp_{ex}(r, b)$
for $(r, b)\in \rreds_{ex} \times \rblues_{ex}$.
This completes the description of the algorithm.

\subsection{Cost analysis}
\label{subsection:grid_cost}
There are two sources of error in our algorithm:
the distortion between $\rinst_\Box$ and $\rinst_{ex}$,
and the error from restricting the solution to the 
internal/external partitioning of demand.

\mparagraph{$\delta$-closeness}
We first formalize the way that $\rinst_\Box$ approximates $\rinst_{ex}$
when it shifts demand to cell centers.
We introduce a notion called $\delta$-closeness between transportation 
instances: informally, two instances are $\delta$-close if we can shift the demands
of one to form the other, without moving any demand more than $\delta$.

For a given transportation map~$\tau$, we define $\tfcost_\infty(\transp) = \max_{(r, b):
\transp(r, b) > 0}{\dist{r}{b}}$ as the maximum distance used in $\transp$.
Now, let $\inst = \tinst{}{}$ be an instance of the transportation problem and $\transp$ a
transportation map for $\inst$. 
Let $\inst' = \tinst{}{'}$ be another instance of the transportation problem
with $\tsupply(\reds) = \tsupply'(\reds')$.
Consider the transportation instances $\inst_\reds = (\reds, \reds', \tsupply_\reds)$
and $\inst_\blues = (\blues, \blues', \tsupply_\blues)$ 
where $\tsupply_\reds$ (resp., $\tsupply_\blues$) is the supply (resp., demand) of points in 
$\reds$ and $\reds'$ (resp., $\blues$ and $\blues'$) in $\inst$ and $\inst'$ 
respectively.
We call $\inst$ and $\inst'$ \emph{$\delta$-close} if there exist
transportation maps $\transp_\reds$ and $\transp_\blues$ of 
$\inst_\reds$ and $\inst_\blues$ such that 
$\tfcost_\infty(\transp_\reds), \tfcost_\infty(\transp_\blues) \leq \delta$.
Observe that this condition holds between the extra supply points and the
contracted points.
\begin{observation}
\label{observation:grid_shift_close}
	Any point in a grid cell of side length $\Delta$ is within $(\Delta/\sqrt{2})$ 
	of the cell center, so $\inst_\Box$ and $\inst_{ex}$ are $(\Delta/\sqrt{2})$-close.
\end{observation}

The definition of $\delta$-closeness states that the supplies of 
$\reds$ (resp., demands of $\blues$) (and therefore the units of $\tau$) 
can be mapped to supplies at $\reds'$ (resp., $\blues'$) within distance $\delta$.
Let $\transp'$ be the result of passing $\transp$ through these maps,
then we say that $\transp'$ solving $\inst'$ is \emph{derived} 
from $\transp$ solving $\inst$.
The following lemma relates the transportation map for $\inst_{ex}$ to the 
solution for $\inst_\Box$ produced by the external subproblem.
\begin{lemma}
\label{lemma:grid_delta_close_error}
	Let $\inst$ and $\inst'$ be two $\delta$-close instances of the 
	transportation problem with $\mathsf{X}'$ being the total demand of each.
	Let $\transp$ be a transportation map of $\inst$ and let $\transp'$ be 
	a transportation map of $\inst'$ derived from $\transp$.
	Then, $|\tfcost(\transp') - \tfcost(\transp)| \leq 2\delta \mathsf{X}'$.
\end{lemma}
\begin{proof}
	Consider a unit of transportation in $\transp$, between points $r, b$.
	In $\transp'$, this unit is mapped to two new points $r', b'$,
	which by $\delta$-closeness must satisfy 
	$\dist{r}{r'} \leq \delta$ and $\dist{b}{b'} \leq \delta$.
	By triangle inequality, $\dist{r'}{b'} \leq \dist{r}{b} + 2\delta$.
	Summing over every unit transported in $\transp$, the lemma follows.
\end{proof}

\mparagraph{Partitioning of demand}
Now that we have quantified the error between $\inst_\Box$ and $\inst_{ex}$,
we analyze the second source of error.
First, we bound the error in a single subproblem (i.e. one subdividing grid),
and then the error for a single pair $(r, b) \in \reds \times \blues$ across 
all subproblems.
We combine these two arguments to bound the expected error
due to the partitioning across all subproblems and all pairs of points.

Fix a recursive problem $\rinst = \rtinst{}{}$, with cell side length $\Delta$.
Let $\excess_\pi$, $\rreds_\pi$, $\rblues_\pi$, $(\rreds_{ex})_\pi$, $(\rblues_{ex})_\pi$ 
for $\pi \in \Pi$ and $\reds_\Box$, $\blues_\Box$, $\reds_{ex}$, $\blues_{ex}$, 
$\rsupply_{ex}$, be as defined in Section~\ref{subsection:grid_high_level}.
Let $\EuScript{I} = \bigcup_{\pi \in \Pi} \rreds_\pi \times \rblues_\pi$
be the set of ``local'' point pairs, solved by the algorithm within 
internal subproblems.
We refer to a pair $(r, b) \in (\rreds \times \rblues) \setminus \EuScript{I}$ 
as ``non-local''.

The next lemma outlines a method for deforming an arbitrary transportation
map to one that respects the local/non-local partitioning used by the 
algorithm.

\begin{lemma}
\label{lemma:grid_swap}
	Let $\inst = \tinst{}{}$ be a recursive subproblem with cell side length $\Delta$,
	and let $\hat{\transp}$ be an arbitrary transportation map for $\inst$.
	Let $\hat{\mathsf{X}} = \sum_{(r, b) \notin \EuScript{I}} \hat{\transp}(r, b)$ 
	be the total non-local transport in $\hat{\transp}$, 
	and $\mathsf{X} = \sum_{\pi \in \Pi} \excess_\pi$ 
	be the total extra supply of $\rinst$.
	Then, there exists a transportation map $\tilde{\transp}$
	comprising local solutions $\tilde{\transp}_\pi$ for each $\rinst_\pi$ 
	and a non-local solution $\tilde{\transp}_{ex}$ for $\rinst_{ex}$,
	such that the following properties hold:
	\begin{enumerate}[A.]
	\item \label{item:grid_swap12}       
		The cost of the transportation map $\tilde{\transp}$ is bounded above
		\begin{equation*}
			\tfcost(\tilde{\transp}) 
			= \sum_{\pi \in \Pi} \tfcost(\tilde{\transp}_\pi) + \tfcost(\tilde{\transp}_{ex})
			\leq \tfcost(\hat{\transp}) + 8\sqrt{2}\Delta\hat{\mathsf{X}}.
		\end{equation*}
	\item \label{item:grid_swap3}
		The local transport in $\tilde{\transp}$ satisfies
		$\tilde{\transp}(r, b) \geq \hat{\transp}(r, b)$ for all $(r, b) \in \EuScript{I}$,
	\item \label{item:grid_swap4}
		and the sum of local difference
		$\sum_{(r, b) \in \EuScript{I}} \left(\tilde{\transp}(r, b) - \hat{\transp}(r, b)\right) \leq 3\hat{\mathsf{X}}$.
	\item \label{item:grid_swap5}
		The non-local transport in $\tilde{\transp}$ satisfies
		$\hat{\mathsf{X}} \geq \mathsf{X} = \tilde{\mathsf{X}} := \sum_{(r, b) \notin \EuScript{I}} \tilde{\transp}(r, b)$.
	\end{enumerate}
\end{lemma}

\begin{proof}
	We deform $\hat{\transp}$ to create a transportation $\tilde{\transp}$ 
	consistent with $\transp$,
	in the sense that the local transport within a cell $\pi$ is between 
	$\rreds_\pi \times \rblues_\pi$, and the non-local transport between 
	different cells is through pairs of $\rreds_{ex} \times \rblues_{ex}$.
	Furthermore, the non-local transport from $\pi$ is exactly $\excess_\pi$.
	Initially, let $\tilde{\transp} = \hat{\transp}$.

	\mparagraph{Stage 1}
	%\subparagraph*{Stage 1.}
	The first stage ensures that the transport within each cell of 
	$\Pi$ is maximal.
	Suppose there is a cell $\pi \in \Pi$ with two points 
	$r_1 \in \rreds \cap \pi$ and $b_2 \in \rblues \cap \pi$,
	such that $\tilde{\transp}(r_1, b_1), \tilde{\transp}(r_2, b_2) > 0$
	where $b_1 \in \rblues \setminus \pi$ and $r_2 \in \rreds \setminus \pi$.
	Let $x = \min\{\tilde{\transp}(r_1, b_1), \tilde{\transp}(r_2, b_2)\}$.
	We redistribute the transport:
	\begin{equation}
	\label{equation:grid_cost2}
	\begin{aligned}
		\tilde{\transp}(r_1, b_1) = \tilde{\transp}(r_1, b_1) - x, \quad
		\tilde{\transp}(r_2, b_2) = \tilde{\transp}(r_2, b_2) - x, \\
		\tilde{\transp}(r_1, b_2) = \tilde{\transp}(r_1, b_2) + x, \quad
		\tilde{\transp}(r_2, b_1) = \tilde{\transp}(r_2, b_1) + x.
	\end{aligned}
	\end{equation}
	This step reduces the non-local transport by at least $x$.
	We repeat this step until there is no cell with
	transport along two pairs $(r_1, b_1)$, $(r_2, b_2)$ as described above.
	When Stage 1 ends, the total non-local transport from each cell 
	$\pi$ is exactly $\excess_\pi$.
	However, some of the local transport within $\pi$ may use 
	nonlocal points ($(\rreds_{ex})_\pi \cup (\rblues_{ex})_\pi$).

	\mparagraph{Stage 2}
	%\subparagraph*{Stage 2.}
	The second stage deforms $\tilde{\transp}$ to 
	ensure that non-local transport is only through the pairs of 
	$\rreds_{ex} \times \rblues_{ex}$.
	Suppose there is a cell $\pi \in \Pi$ with $(\rreds_{ex})_\pi \neq \emptyset$
	and two points $r_1 \in \reds_\pi, r_2 \in (\rreds_{ex})_\pi$ such that
	$\tilde{\transp}(r_1, b_1), \tilde{\transp}(r_2, b_2) > 0$,
	where $b_1 \not\in B_\pi$ and $b_2 \in B_\pi$.
	We perform the same redistribution as (\ref{equation:grid_cost2}) 
	between this new choice of $r_1, b_1, r_2, b_2$.

	\begin{figure}
	\centering
	\subcaptionbox{\centering \label{subfigure:grid_swap1}}{\includegraphics[width=0.4\linewidth,page=1]{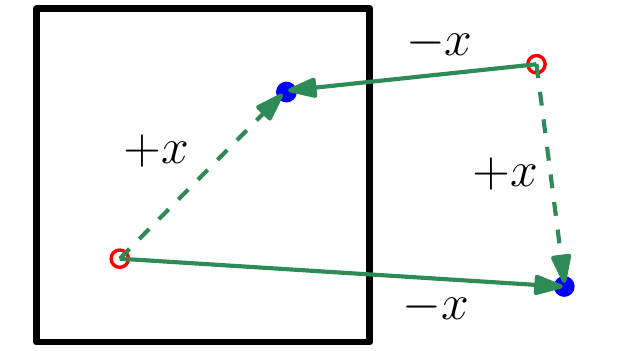}}
	\subcaptionbox{\centering \label{subfigure:grid_swap2}}{\includegraphics[width=0.4\linewidth,page=2]{fig/swap.pdf}}
	\caption{Stage 1 (a) and Stage 2 (b). The square-shaped point was in $(\rreds_{ex})_\pi$.}
	\label{figure:grid_swap}
	\end{figure}

	When Stage 2 terminates, non-local transport 
	is only through nonlocal pairs ($(\rreds_{ex}) \times (\rblues_{ex})$) and 
	local transport in $\pi$ is between local pairs ($\rreds_\pi \times \rblues_\pi$).
	Let $\tilde{\transp}_\pi$ be the restriction of $\tilde{\transp}$ to
	the pairs in $\rreds_\pi \times \rblues_\pi$ for $\pi \in \Pi$,
	and let $\tilde{\transp}_{ex}$ be the restriction of $\tilde{\transp}$
	to pairs in $\rreds_{ex} \times \rblues_{ex}$.
	Then $\tilde{\transp}_\pi$ is a transportation map of $\rinst_\pi$
	and $\tilde{\transp}_{ex}$ is a transportation map of $\rinst_{ex}$.
	During the procedure, we never reduce $\tilde{\transp}(r, b)$ 
	if $(r, b)$ is a local edge; this proves \ref{item:grid_swap3}.

	We note that $\mathsf{X} \leq \hat{\mathsf{X}}$. 
	The execution of (\ref{equation:grid_cost2}) increases the cost 
	of transport by at most $4\sqrt{2}\Delta x$ 
	(by triangle inequality, see Figure~\ref{figure:grid_swap})
	and reduces the non-local transport by at least $x$.
	Hence, the total increase in the cost of the transport after 
	Stage 1 is at most $4\sqrt{2}\Delta\hat{\mathsf{X}}$.
	A similar argument shows that Stage 2 increases the cost by at most 
	$4\sqrt{2}\Delta \mathsf{X} \leq 4\sqrt{2} \Delta \hat{\mathsf{X}}$.
	Hence,
	\begin{equation*}
		\tfcost(\tilde{\transp}) \leq \tfcost(\hat{\transp}) + 8\sqrt{2}\Delta\hat{\mathsf{X}}.
	\end{equation*}
	Furthermore, the execution of (\ref{equation:grid_cost2})
	increases the local transport by at most $2x$ in Stage 1 and 
	by $x$ in Stage 2.
	Hence, total increase in local transport during the deformation 
	is at most $3\hat{\mathsf{X}}$. 
	Therefore,
	\begin{equation*}
		\sum_{(r, b) \in \EuScript{I}} \left(\tilde{\transp}(r, b) - \hat{\transp}(r, b)\right) \leq 3\hat{\mathsf{X}}. \qedhere
	\end{equation*}
\end{proof}

\mparagraph{Error parameter $\cutlen$}
In the previous lemma, we bounded the error due to a single subproblem.
We now bound the error due to a single pair of points 
$(r, b) \in \reds \times \blues$, using a random variable $\cutlen(r, b)$,
defined as the cell side length of the first recursive grid to split 
$(r, b)$ into different cells.

Formally, recall that a recursive subproblem may split a point 
$p \in \reds \cup \blues$ into two copies $p'$ and $p''$ with 
$\tsupply(p') + \tsupply(p'') = \tsupply(p)$;
one of them passed to the external subproblem,
and the other passed down to an internal subproblem.
Abusing notation slightly, we use $\reds$ and $\blues$ to denote the 
multisets that contain all copies of points that are split along 
with the updated demands.

For any base subproblem $\rtinst{}{_\mathrm{base}}$,
every point $p \in \rreds_\mathrm{base} \cup \rblues_\mathrm{base}$ can be identified with a point
$p \in \reds \cup \blues$ such that $\rsupply_\mathrm{base}(p) = \tsupply(p)$.
With this interpretation, we define a function 
$\cutlen: \reds \times \blues \to \reals_{\geq 0}$ as follows:
If there is a base subproblem $\rtinst{}{_\mathrm{base}}$ such that 
$(r, b) \in \rreds_\mathrm{base} \times \rblues_\mathrm{base}$, we set $\cutlen(r, b) = 0$.
Otherwise, there is a recursive subproblem $\rinst$ where
$(r, b) \in \rreds \times \rblues$, but $r$ and $b$ are split into 
different cells of the randomly shifted grid.
In this case, $\cutlen(r, b)$ denotes the side length of the grid cells, 
i.e. $\cutlen(r, b) = \ell/m^\delta$ where $\ell$ is the length of the 
smallest square containing $\rreds \cup \rblues$,
and $m = |\rreds \cup \rblues|$.

The next lemma bounds the expected value of the error parameter
$\cutlen(r, b)$ in terms of the distance $\dist{r}{b}$ 
for any pair $(r, b) \in \reds \times \blues$. 
Its proof uses our choice of safe grids to argue that,
though the recursion depth can be large, 
the number of recursive subproblems that can potentially split $(r,b)$ is small.

\begin{lemma}
\label{lemma:grid_expected_cut}
	There exists a constant $c_1 > 0$ such that for any 
	$(r, b) \in \reds \times \blues$, the expectation
	$\mathsf{E}\left[\cutlen(r, b)\right] \leq c_1 \log_2(1/\eps) \dist{r}{b}$. 
\end{lemma}
\begin{proof}
	We say a recursive subproblem $\inst = \tinst{}{}$ is \emph{relevant} if contains $r$ and $b$
	and there is some safe random shift that splits $(r,b)$.
	Suppose the algorithm creates some relevant recursive subproblem $\inst$ over $m$ points, and
	let $\ell$ be the side length of the smallest orthogonal bounding square of that subproblem.
	The measure of each (safe) horizontal and vertical random shift
	is at least
	\begin{equation*}
		\frac{\ell}{m^\delta} - \frac{2\ell}{m^3} \cdot m 
			= \left(1 - \frac{2}{m^{2-\delta}}\right) \cdot \frac{\ell}{m^\delta}
			\geq \left(1 - \frac{1}{m}\right) \cdot \frac{\ell}{m^\delta},
	\end{equation*}
	whereas the measure of horizontal/vertical shifts 
	that split $(r, b)$ is at most $\dist{r}{b}$.
	Hence, the probability that a safe random shift of this 
	grid splits $(r, b)$ is at most
	\begin{equation*}
		\frac{2 \dist{r}{b}}{(1 - 1/m_i) (\ell/m^\delta)} 
			\leq 3 \dist{r}{b} \cdot \frac{m^\delta}{\ell}.
	\end{equation*}
	We say the \emph{expected contribution} of~$\inst$ to $\cutlen(r,b)$ is
	\begin{align*}
		\Pr[\text{grid for $\inst$ splits $(r,b)$}] \cdot \frac{\ell}{m^\delta} 
			&\leq 3 \dist{r}{b} \cdot \frac{m^\delta}{\ell} \cdot \frac{\ell}{m^\delta} \\
			&= 3 \dist{r}{b}.
	\end{align*}

	Now, let~$q_i$ be the probability that the algorithm creates at least~$i$ relevant recursive
	subproblems.
	We have
	\begin{equation*}
		\mathsf{E}\left[\cutlen(r,b)\right] 
			\leq \sum_{i \geq 1} q_i \cdot 3 \dist{r}{b}
			\leq \sum_{i \geq 1 : q_i > 0} 3\dist{r}{b}
	\end{equation*}
	Therefore, we may prove the lemma by bounding the number of relevant recursive subproblems that
	arise during a single run of the algorithm by $O(\log(1/\eps))$.
	To do so, we will categorize these subproblems in such a way that there are
	$O(\log(1/\eps))$ categories and each category can contain at most one relevant subproblem
	during a single run of the algorithm.

	Recall, $n$ is the number of points in the input instance.
	%and $m$ be the number of points in the subproblem which splits $(r, b)$,
	%and $\ell$ the side length of the smallest orthogonal bounding square
	%of that subproblem (i.e. $\cutlen(r, b) = \ell/m^\delta$).
	Again, suppose the algorithm creates some relevant recursive subproblem $\inst$ over $m$ points,
	and let $\ell$ be the side length of the smallest orthogonal bounding square of that subproblem.
	We may assume that $m > n^{\eps/4}$, since otherwise the subproblem is a base case and
	therefore not relevant.

	Define $\underline{\ell} := \dist{r}{b} / \sqrt{2}$; 
	clearly, $\ell \geq \underline{\ell}$.
	We partition the interval $[n^{\eps/4}, n]$ into 
	$u = \lceil\log_2 (4/\eps)\rceil$ intervals of the form 
	$[n_j, n_j^2]$, where $n_j = n^{2^{-j}}$, for $1 \leq j \leq u$.
	There exists an index $j^*$ where $m \in [n_{j^*}, n_{j^*}^2]$, 
	and $\ell \leq \overline{\ell} := (n_{j^*}^2)^3 \dist{r}{b}$
	because the grid is safe.
	Thus, $\ell \in [\underline{\ell}, \overline{\ell}]$,
	where $\overline{\ell} / \underline{\ell} = \sqrt{2} n_{j^*}^6$.

	The interval $[\underline{\ell}, \overline{\ell}]$ can be 
	covered by $7/\delta > (1/\delta)(6 + \log_2(\sqrt{2})/\log_2(n_{j^*}))$ 
	intervals of the form 
	$J_i = [n^{i\delta}_{j^*} \underline{\ell}, n^{(i+1)\delta}_{j^*} \underline{\ell}]$ 
	for $0 \leq i \leq 7/\delta = O(1)$.
	For each value of $i,j$, the algorithm produces at most one subproblem containing a number of
	points in $[n_{j}, n_{j}^2]$ and bounding square side length in $J_i$.
	The total number of intervals is no more than 
	$\lceil\log_2 (4/\eps)\rceil \cdot 7/\delta = O(\log(1/\eps))$.
\end{proof}

\mparagraph{Expected cost of algorithm}
We are now ready to analyze the expected cost of $\transp$, 
the algorithm's transportation map.
First, we analyze the cost if $\inst$ has bounded spread.
Recall that, in this case, our algorithm computes an optimal solution
for each external subproblem (using Orlin's algorithm).

\begin{lemma}
\label{lemma:grid_expected_cost}
	If $\inst$ is a transportation instance with bounded spread,
	then there exists a constant $c_2 > 0$ such that for any 
	transportation map $\hat{\transp}$ of $\inst$,
	\begin{equation*}
		\tfcost(\transp) \leq \tfcost(\hat{\transp}) + c_2 \sum_{(r, b) \in \reds \times \blues} \hat{\transp}(r, b) \cutlen(r, b).
	\end{equation*}
\end{lemma}
An immediate corollary of Lemmas \ref{lemma:grid_expected_cut} and \ref{lemma:grid_expected_cost} is:
\begin{corollary}
\label{corollary:grid_approx_bs}
	If $\inst$ has bounded spread, then 
	$\mathsf{E}[\tfcost(\transp)] = O(\log (1/\eps)) \tfcost(\transp^*)$.
\end{corollary}
\begin{proof}[Proof of Lemma~\ref{lemma:grid_expected_cost}]
	We prove the lemma by induction on the number of points in the subproblem.
	If $\inst$ is a base problem, then $\transp$ is an optimal transport 
	of $\inst$ and the lemma holds.
	Otherwise $\Box$, the smallest square containing $\reds \cup \blues$,
	is split into a set of grid cells.
	Following the notation in Section~\ref{subsection:grid_high_level},
	let $\Pi$ be the set of non-empty cells
	and $\Delta$ the side length of each grid cell.

	Recall that $\transp$ is the combination of solutions $\transp_\pi$ 
	for the internal subproblems
	$\inst_\pi = (\reds_\pi, \blues_\pi, \tsupply_\pi)$ of $\pi \in \Pi$, 
	and the map $\transp_{ex}$
	for $\inst_{ex} = (\reds_{ex}, \blues_{ex}, \tsupply_{ex})$ derived from the 
	solution $\transp_\Box$ to the external subproblem 
	$\inst_\Box = (\reds_\Box, \blues_\Box, \tsupply_\Box)$.
	From Observation~\ref{observation:grid_shift_close}, 
	$\inst_\Box$ and $\inst_{ex}$ are $(\Delta/\sqrt{2})$-close;
	thus Lemma~\ref{lemma:grid_delta_close_error} implies
	$\tfcost(\transp_\Box) \leq \tfcost(\transp_{ex}) + \sqrt{2} \Delta \mathsf{X}$.
	We have
	\begin{equation}
	\label{equation:grid_cost1}
		\tfcost(\transp) 
		= \tfcost(\transp_{ex})
			+ \sum_{\pi \in \Pi} \tfcost(\transp_\pi) 
		\leq \tfcost(\transp_\Box) 
			+ \sqrt{2} \Delta \mathsf{X}
			+ \sum_{\pi \in \Pi} \tfcost(\transp_\pi).
	\end{equation}
	Thus, using (\ref{equation:grid_cost1}), we can bound
	$\transp$ by bounding the local ($\transp_\pi$) and 
	non-local ($\transp_\Box$) solutions individually.

	Let $\tilde{\transp}$ be the transportation map created by deforming
	$\hat{\transp}$ in Lemma~\ref{lemma:grid_swap},
	with $\tilde{\transp}_\pi$ and $\tilde{\transp}_{ex}$
	its restrictions to local and non-local pairs of points respectively.
	To bound $\transp_\Box$, notice that that $\tilde{\transp}_{ex}$ 
	solves $\inst_{ex} = \tinst{}{_{ex}}$, and $\inst_{ex}$ is 
	$(\Delta/\sqrt{2})$-close to $\inst_\Box$.
	We apply Lemma~\ref{lemma:grid_delta_close_error} and 
	optimality of $\transp_\Box$ to conclude,
	\begin{equation}
	\label{equation:grid_cost6}
		\tfcost(\transp_\Box)
			\leq \tfcost(\tilde{\transp}_{ex}) + \sqrt{2}\Delta \mathsf{X}.
	\end{equation}
	We now bound the local solutions.
	Since $\tilde{\transp}_\pi$ is a transportation map of the 
	internal subproblem $\inst_\pi$, by the induction hypothesis,
	\begin{equation}
	\label{equation:grid_cost7}
		\tfcost(\transp_\pi) \leq 
		\tfcost(\tilde{\transp}_\pi)
		+ c_2 \sum_{(r, b) \in \reds_\pi \times \blues_\pi} \tilde{\transp}_\pi(r, b) \cutlen(r, b).
	\end{equation}
	We can now combine (\ref{equation:grid_cost6}) and 
	(\ref{equation:grid_cost7}) to bound $\transp$ in (\ref{equation:grid_cost1}). 
	\begin{align*}
		\tfcost(\transp)
		&\leq \tfcost(\transp_\Box) 
			+ \sum_{\pi \in \Pi} \tfcost(\transp_\pi) 
			+ \sqrt{2}\Delta\mathsf{X} 
			& \\
		&\leq \tfcost(\tilde{\transp}_{ex}) 
			+ 2\sqrt{2}\Delta\mathsf{X}
			+ \sum_{\pi \in \Pi} 
				\left[ \tfcost(\tilde{\transp}_\pi) 
				+ c_2 \sum_{\pi \in \Pi} \sum_{(r, b) \in \reds_\pi \times \blues_\pi} \tilde{\transp}_\pi(r, b) \cutlen(r, b) 
				\right] 
			&\text{(by \eqref{equation:grid_cost6}, \eqref{equation:grid_cost7})} \\
		&= \tfcost(\tilde{\transp}) 
			+ c_2 \sum_{(r, b) \in \EuScript{I}} \tilde{\transp}(r, b) \cutlen(r, b)
			+ 2\sqrt{2}\Delta\mathsf{X}
			& \text{(Lem.~\ref{lemma:grid_swap}\ref{item:grid_swap12})} \\
		&\leq \tfcost(\hat{\transp})
			+ c_2 \sum_{(r, b) \in \EuScript{I}} \tilde{\transp}(r, b) \cutlen(r, b)
			+ 10\sqrt{2}\Delta\hat{\mathsf{X}}
			& \text{(Lem.~\ref{lemma:grid_swap}\ref{item:grid_swap12}, \ref{lemma:grid_swap}\ref{item:grid_swap5})} \\
		&= \tfcost(\hat{\transp})
			+ c_2 \sum_{(r, b) \in \reds \times \blues} \hat{\transp}(r, b) \cutlen(r, b)
			+ \Gamma,
			&
	\end{align*}
		$$ \text{where~} \Gamma = c_2 \sum_{(r, b) \in \EuScript{I}} \left(\tilde{\transp}(r, b) - \hat{\transp}(r, b)\right) \cutlen(r, b)
			+ 10\sqrt{2}\Delta\hat{\mathsf{X}}
			- c_2 \sum_{(r, b) \notin \EuScript{I}} \hat{\transp}(r, b) \cutlen(r, b).$$
	By definition, $\cutlen(r, b) = \Delta$ for $(r, b) \not\in \EuScript{I}$,
	$\cutlen(r, b) \leq \Delta/(n^{\eps/4})^\delta \leq \Delta/4$ for $(r, b) \in \EuScript{I}$,
	and $\sum_{(r, b) \notin \EuScript{I}} \hat{\transp}(r, b) = \hat{\mathsf{X}}$.
	Therefore, using Lemma~\ref{lemma:grid_swap}\ref{item:grid_swap4}, 
	\begin{equation*}
		\Gamma \leq c_2 \frac{\Delta}{4} \cdot 3 \hat{\mathsf{X}} + 10\sqrt{2}\Delta\hat{\mathsf{X}} - c_2\Delta\hat{\mathsf{X}}
			= \left(10 \sqrt{2} - \frac{c_2}{4}\right) \Delta \hat{\mathsf{X}} \leq 0.
	\end{equation*}
	provided that $c_2 \geq 40\sqrt{2}$.
	Hence,
	$\tfcost(\transp) \leq \tfcost(\hat{\transp}) + c_2 \sum_{(r, b)} \hat{\transp}(r, b) \cutlen(r, b)$.
	This completes the proof of the lemma.
\end{proof}

\mparagraph{The general case}
%\subparagraph*{The general case.}
We now analyze the cost of the transportation map for the general case,
when the spread of $\reds \cup \blues$ is arbitrary.

Recall the categorization of recursive subproblems as primary, secondary, 
or tertiary based on the number of external subproblem invocations on 
its path in the recursion tree (see Figure~\ref{figure:grid_subproblems}).
We now introduce two functions 
$\cutlen_1, \cutlen_2: \reds \times \blues \to \reals_{\geq 0}$
corresponding to the errors introduced in the primary and secondary recursions.
(Note that tertiary subproblems are solved exactly; hence, there is 
no error introduced in solving a tertiary subproblem.)

The function $\cutlen_1$ corresponds to the primary recursion and is 
the same as the $\cutlen$ defined before, i.e., $\cutlen_1(r, b) = 0$ 
if $(r, b)$ belongs to a primary base subproblem, 
otherwise it is the length of the grid cell at the subproblem which 
splits $r$ and $b$.
The function $\cutlen_2(r, b)$ corresponds to the secondary recursion for 
$(r, b)$. 
If $(r, b)$ belongs to a primary base problem 
(i.e. does not appear in any secondary recursion), 
then we set $\cutlen_2(r, b) = 0$.
Otherwise, let $\bar{r} \in \reds_\Box$ and $\bar{b} \in \blues_\Box$ 
be the centers of the grid cells of the primary subproblem where 
$r$ and $b$ were split. Then, $\cutlen_2(r, b)$ is defined 
to be $\cutlen(\bar{r}, \bar{b})$ for 
the secondary recursion on $(\reds_\Box, \blues_\Box)$.
From this definition, we observe that:

\begin{lemma}
\label{lemma:grid_cutlen_prop}
	\begin{enumerate}[A.]
	\item \label{item:grid_cutlen_prop1}
		$\cutlen_2(r, b) = \cutlen(r_{\pi_1}, b_{\pi_2})$, for $(r, b) \in (\reds_{ex})_{\pi_1} \times (\blues_{ex})_{\pi_2}$
	\item \label{item:grid_cutlen_prop2}
		$\cutlen_1(r, b) = \Delta$ for $(r, b) \in \reds_{ex} \times \blues_{ex}$
	\item \label{item:grid_cutlen_prop3}
		$\cutlen_1(r, b), \cutlen_2(r, b) \leq \Delta/(n^{\eps/4})^\delta \leq \Delta/12$ for $(r, b) \in \reds_\pi \times \blues_\pi$
	\end{enumerate}
\end{lemma}

We now state two lemmas that are counterparts of Lemmas 
\ref{lemma:grid_expected_cut} and \ref{lemma:grid_expected_cost}
for the general case.

\begin{lemma}
\label{lemma:grid_expected_cut_gen}
	For any pair $(r, b) \in \reds \times \blues$,
	$\mathsf{E}\left[\cutlen_2(r, b)\right] = O(\log^2 (1/\eps)) \dist{r}{b}$.
\end{lemma}
\begin{proof}
	Let $\bar{r}, \bar{b}$ be the grid centers as defined earlier; then 
	$\dist{\bar{r}}{\bar{b}} \leq \dist{r}{b} + 2\sqrt{2} \cutlen_1(r, b)/2$.
	We first prove the expectation of $\cutlen_2(r, b)$ 
	conditioned on the value of $\cutlen_1(r, b)$.
	By Lemma~\ref{lemma:grid_expected_cut},
	\begin{equation*}
		\mathsf{E}\left[\cutlen_2(r, b) \mid \cutlen_1(r, b)\right]
			= O(\log (1/\eps)) \dist{\bar{r}}{\bar{b}}
			\leq O(\log (1/\eps)) \left(\dist{r}{b} + \sqrt{2}\cutlen_1(r, b)\right).
	\end{equation*}
	By applying Lemma~\ref{lemma:grid_expected_cut} again, we obtain
	\begin{equation*}
		\mathsf{E}\left[\cutlen_2(r, b)\right]
			\leq O(\log (1/\eps)) \left[ \dist{r}{b} + \mathsf{E}\left[\cutlen_1(r, b)\right] \right]
			\leq O(\log^2 (1/\eps)) \dist{r}{b}. \qedhere
	\end{equation*}
\end{proof}

\begin{lemma}
\label{lemma:grid_expected_cost_gen}
	There exists a constant $c_3 > 0$ such that for any 
	transportation map $\hat{\transp}$ of $\inst$,
	\begin{equation*}
		\tfcost(\transp) \leq \tfcost(\hat{\transp}) + c_3 \sum_{(r, b) \in \reds \times \blues} \hat{\transp}(r, b) \left[\cutlen_1(r, b) + \cutlen_2(r, b)\right].
	\end{equation*}
\end{lemma}

The bound on the expected cost follows directly from the above two lemmas.

\begin{corollary}
\label{corollary:grid_approx_gen}
	$\mathsf{E}[\tfcost(\transp)] = O(\log^2 (1/\eps)) \tfcost(\transp^*)$.
\end{corollary}

We now prove Lemma~\ref{lemma:grid_expected_cost_gen}, our main technical lemma.

\begin{proof}[Proof of Lemma~\ref{lemma:grid_expected_cost_gen}.]
	The proof is similar to Lemma~\ref{lemma:grid_expected_cost}, 
	except that $\transp_\Box$ is not an optimal solution of the external subproblem 
	$\inst_\Box = (\reds_\Box, \blues_\Box, \tsupply_\Box)$.
	Instead $\transp_\Box$ is computed recursively.
	Let $\inst_{ex} = \tinst{}{_{ex}}$ be the transportation instance 
	as defined in the proof of Lemma~\ref{lemma:grid_expected_cost},
	which is $(\Delta/\sqrt{2})$-close to $\inst_\Box$.
	Recall $\mathsf{X} = \sum_{\pi \in \Pi} \excess_\pi$ 
	is the total extra supply of $\inst$, and 
	$\hat{\mathsf{X}} = \sum_{(r, b) \notin \EuScript{I}} \hat{\transp}(r, b)$ 
	is total non-local transport in $\hat{\transp}$.
	We will first bound $\tfcost(\transp_\Box)$, then each $\tfcost(\transp_\pi)$.
	See Figure~\ref{figure:grid_cost_gen_notation}.
    
\begin{figure}
	\centering
	\includegraphics[width=0.5\linewidth]{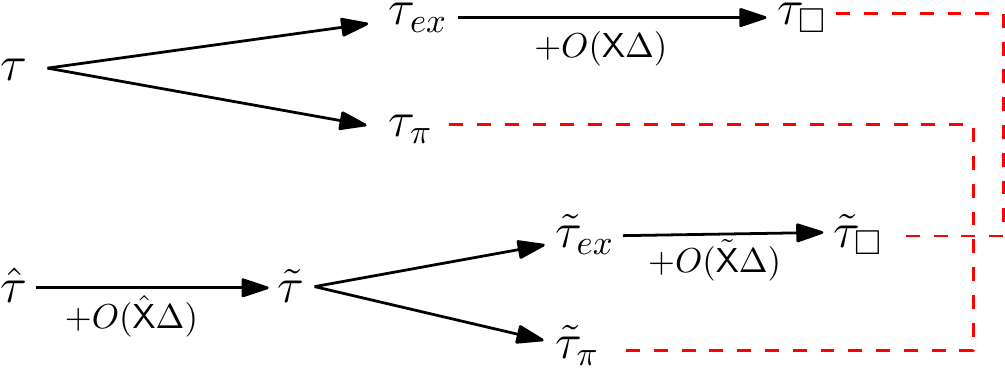}
	\caption{
		The different transportation maps used in the cost proof. 
		The dotted lines indicate the maps that are compared.
	}
	\label{figure:grid_cost_gen_notation}
\end{figure}

	Let $\tilde{\transp}_{ex}$ be a transportation map of $\inst_{ex}$ constructed from
	the deformation of $\hat{\transp}$ as in Lemma~\ref{lemma:grid_swap}.
	Let $\tilde{\transp}_\Box$ be the projection of $\tilde{\transp}_{ex}$ 
	onto the merged points of $\inst_\Box$,
	i.e. $\tilde{\transp}_\Box$ is a transportation map for $\inst_\Box$ 
	derived from $\tilde{\transp}_{ex}$.
	We begin by applying Lemma~\ref{lemma:grid_expected_cost}
	to $\transp_\Box$, since it has bounded spread.
	\begin{equation}
	\label{equation:grid_cost_gen6}
	\begin{aligned}
		\tfcost(\transp_\Box) 
		&\leq \tfcost(\tilde{\transp}_\Box) 
			+ c_2 \sum_{(r_{\pi_1}, b_{\pi_2}) \in \reds_\Box \times \blues_\Box} \tilde{\transp}_\Box(r_{\pi_1}, b_{\pi_2}) \cutlen(r_{\pi_1}, b_{\pi_2})
			& \text{(by L\ref{lemma:grid_expected_cost})} \\
		&\leq \tfcost(\tilde{\transp}_{ex}) 
			+ \sqrt{2}\Delta\mathsf{X}
			+ c_2 \sum_{\substack{(r_{\pi_1}, b_{\pi_2})\\ \in \reds_\Box \times \blues_\Box}} 
				\tilde{\transp}_\Box(r_{\pi_1}, b_{\pi_2}) \cutlen(r_{\pi_1}, b_{\pi_2})
			& \text{(by L\ref{lemma:grid_delta_close_error})} \\
		&= \tfcost(\tilde{\transp}_{ex}) 
			+ \sqrt{2}\Delta\mathsf{X}
			+ c_2 \sum_{\substack{(r_{\pi_1}, b_{\pi_2})\\ \in \reds_\Box \times \blues_\Box}} \cutlen(r_{\pi_1}, b_{\pi_2}) 
			\sum_{\substack{(r, b)\\ \in (\reds_{ex})_{\pi_1} \times (\blues_{ex})_{\pi_2}}} \tilde{\transp}_{ex}(r, b)
			& \\
		&= \tfcost(\tilde{\transp}_{ex}) 
			+ \sqrt{2}\Delta\mathsf{X}
			+ c_2 \sum_{(r, b) \in \reds_{ex} \times \blues_{ex}} \tilde{\transp}_{ex}(r, b) \cutlen_2(r, b).
			& \text{(by L\ref{lemma:grid_cutlen_prop}.\ref{item:grid_cutlen_prop1})}
	\end{aligned}
	\end{equation}
	This completes our bound on $\tfcost(\transp_\Box)$.
	For each $\transp_\pi$, we begin with the inductive hypothesis. 
	\begin{equation*}
		\tfcost(\transp_\pi) 
		\leq \tfcost(\tilde{\transp}_\pi) 
			+ c_3 \sum_{(r, b) \in \reds_\pi \times \blues_\pi} \tilde{\transp}_\pi(r, b) [\cutlen_1(r, b) + \cutlen_2(r, b)] 
	\end{equation*}
	In the next few steps, we replace $\tilde{\transp}_\pi$ terms with $\hat{\transp}_\pi$ 
	at the cost of $O(\Delta\hat{\mathsf{X}})$ over all $\pi \in \Pi$.
	\begin{align*}
		\sum_{\pi \in \Pi} \tfcost(\transp_\pi) 
		&\leq \sum_{\pi \in \Pi} \left(
			\tfcost(\tilde{\transp}_\pi)
			+ c_3 \sum_{(r, b) \in \reds_\pi \times \blues_\pi} \tilde{\transp}_\pi(r, b) [\cutlen_1(r, b) + \cutlen_2(r, b)] 
			\right)
			\\
		&= \sum_{\pi \in \Pi} 
			\!\begin{aligned}[t]
			\Bigg( % unfortunately, \left( \right) doesn't work across lines
			\tfcost(\tilde{\transp}_\pi)
			&+ c_3 \sum_{(r, b) \in \reds_\pi \times \blues_\pi} \hat{\transp}_\pi(r, b) [\cutlen_1(r, b) + \cutlen_2(r, b)] \\
			&+ c_3 \sum_{(r, b) \in \reds_\pi \times \blues_\pi} (\tilde{\transp}_\pi(r, b) - \hat{\transp_\pi}(r, b)) [\cutlen_1(r, b) + \cutlen_2(r, b)]
			\Bigg)
			\end{aligned} 
			\\
		&\leq 
			\!\begin{aligned}[t]
			&\sum_{\pi \in \Pi} \left(
				\tfcost(\tilde{\transp}_\pi)
				+ c_3 \sum_{(r, b) \in \reds_\pi \times \blues_\pi} \hat{\transp}_\pi(r, b) [\cutlen_1(r, b) + \cutlen_2(r, b)]
				\right) \\
			&+ c_3 \sum_{(r, b) \in \EuScript{I}} (\tilde{\transp}_\pi(r, b) - \hat{\transp}_\pi(r, b))[\cutlen_1(r, b) + \cutlen_2(r, b)] \\
			\end{aligned}
	\end{align*}
	The second term can be reduced using
	Lemma~\ref{lemma:grid_swap}\ref{item:grid_swap4} and
	Lemma~\ref{lemma:grid_cutlen_prop}.\ref{item:grid_cutlen_prop3}.
	\begin{equation*}
		c_3 \sum_{(r, b) \in \EuScript{I}} (\tilde{\transp}_\pi(r, b) - \hat{\transp_\pi}(r, b))[\cutlen_1(r, b) + \cutlen_2(r, b)]
		\leq c_3 \cdot (3\hat{\mathsf{X}}) (2\Delta / 12)
		= \frac{c_3}{2}\Delta\hat{\mathsf{X}}
	\end{equation*}
	Together, we have a bound on the $\transp_\pi$ cost.
	\begin{equation}
	\label{equation:grid_cost_gen7}
		\sum_{\pi \in \Pi} \tfcost(\transp_\pi) 
		\leq \sum_{\pi \in \Pi} \left(
			\tfcost(\tilde{\transp}_\pi)
			+ c_3 \sum_{(r, b) \in \reds_\pi \times \blues_\pi} \hat{\transp}_\pi(r, b) [\cutlen_1(r, b) + \cutlen_2(r, b)] 
		\right)
		+ \frac{c_3}{2}\Delta\hat{\mathsf{X}}
	\end{equation}
	As before, we combine the bounds (\ref{equation:grid_cost_gen6}) and 
	(\ref{equation:grid_cost_gen7}) on $\tfcost(\transp_\Box)$ 
	and $\tfcost(\transp_\pi)$ respectively.
	Starting from (\ref{equation:grid_cost1}),
	\begin{align*}
		\tfcost(\transp)
		&\leq \tfcost(\transp_\Box) + \sqrt{2}\Delta\mathsf{X} 
			+ \sum_{\pi \in \Pi} \tfcost(\transp_\pi)
			&\\
		&\leq \tfcost(\tilde{\transp}_{ex}) 
			\!\begin{aligned}[t]
			&+ c_2 \sum_{(r, b) \in \reds_{ex} \times \blues_{ex}} \tilde{\transp}_{ex}(r, b) \cutlen_2(r, b)
			+ \left(2\sqrt{2}\Delta\mathsf{X} + \frac{c_3}{2}\Delta\hat{\mathsf{X}} \right) \\
			&+ \sum_{\pi \in \Pi}\left(\tfcost(\tilde{\transp}_\pi) 
				+ c_3 \sum_{(r, b) \in \reds_\pi \times \blues_\pi} \hat{\transp}_\pi(r, b) [\cutlen_1(r, b) + \cutlen_2(r, b)]\right) \\
			\end{aligned} 
			&\\
		&= \tfcost(\tilde{\transp}) 
			\!\begin{aligned}[t]
			&+ c_2 \sum_{(r, b) \notin \EuScript{I}} \tilde{\transp}(r, b) \cutlen_2(r, b)
			+ \left(2\sqrt{2}\Delta\mathsf{X} + \frac{c_3}{2}\Delta\hat{\mathsf{X}} \right) \\
			&+ c_3\sum_{(r, b) \in \EuScript{I}} \hat{\transp}(r, b) [\cutlen_1(r, b) + \cutlen_2(r, b)]
			\end{aligned}
			&\\
		&\leq \tfcost(\hat{\transp}) 
			\!\begin{aligned}[t]
			&+ c_2 \sum_{(r, b) \notin \EuScript{I}} \tilde{\transp}(r, b) \cutlen_2(r, b)
			+ \left(10\sqrt{2}\Delta\hat{\mathsf{X}} + \frac{c_3}{2}\Delta\hat{\mathsf{X}} \right) \\
			&+ c_3\sum_{(r, b) \in \EuScript{I}} \hat{\transp}(r, b) [\cutlen_1(r, b) + \cutlen_2(r, b)]
			%&
			\end{aligned}
			& \text{(by L\ref{lemma:grid_swap}\ref{item:grid_swap12})}
	\end{align*}
	Next, we subtract and replace $\tilde{\transp}$ with $\hat{\transp}$ in the second term.
	\begin{align*}
		\tfcost(\transp) &\leq \tfcost(\hat{\transp}) 
			\!\begin{aligned}[t]
			&+ c_2\sum_{(r, b) \notin \EuScript{I}} \hat{\transp}(r, b) \cutlen_2(r, b)
			+ \left(10\sqrt{2}\Delta\hat{\mathsf{X}} + \frac{c_3}{2}\Delta\hat{\mathsf{X}} \right) \\
			&+ c_3\sum_{(r, b) \in \EuScript{I}} \hat{\transp}(r, b) [\cutlen_1(r, b) + \cutlen_2(r, b)] \\
			&+ c_2\sum_{(r, b) \notin \EuScript{I}} (\tilde{\transp}(r, b) - \hat{\transp}(r, b)) \cutlen_2(r, b)
			\end{aligned}
			\\
		&= \tfcost(\hat{\transp}) 
			\!\begin{aligned}[t]
			&+ c_2\sum_{(r, b) \notin \EuScript{I}} \hat{\transp}(r, b) \cutlen_2(r, b)
			+ \left(10\sqrt{2}\Delta\hat{\mathsf{X}} + \frac{c_3}{2}\Delta\hat{\mathsf{X}} \right) \\
			&+ c_3\sum_{(r, b) \in \EuScript{I}} \hat{\transp}(r, b) [\cutlen_1(r, b) + \cutlen_2(r, b)]
			+ \Gamma,
			\end{aligned}
			\\
		\text{where\quad}& 
		\Gamma = c_2 \sum_{(r, b) \notin \EuScript{I}} (\tilde{\transp}(r, b) - \hat{\transp}(r, b)) \cutlen_2(r, b).
	\end{align*}

	Recall that $\hat{\mathsf{X}}$ is exactly the non-local transportation of $\hat{\transp}$.
	Using Lemma~\ref{lemma:grid_cutlen_prop}.\ref{item:grid_cutlen_prop2},
	\begin{equation*}
		\left(10\sqrt{2}\Delta\hat{\mathsf{X}} + \frac{c_3}{2}\Delta\hat{\mathsf{X}} \right) 
		= (c_2 + 10\sqrt{2} + \frac{c_3}{2}) \sum_{(r, b) \notin \EuScript{I}} \hat{\transp}(r, b) \cutlen_1(r, b) , 
	\end{equation*}
	which we can merge into the non-local sum.
	Then,
	\begin{align*}
		\tfcost(\transp) 
		&\leq \tfcost(\hat{\transp}) 
			\!\begin{aligned}[t]
			&+ \left(c_2 + 10\sqrt{2} + \frac{c_3}{2}\right) 
				\sum_{(r, b) \notin \EuScript{I}} \hat{\transp}(r, b) [\cutlen_1(r, b) + \cutlen_2(r, b)] \\
			&+ c_3\sum_{(r, b) \in \EuScript{I}} \hat{\transp}(r, b) [\cutlen_1(r, b) + \cutlen_2(r, b)]
			+ \Gamma
			\end{aligned}
			\\
		&\leq \tfcost(\hat{\transp}) 
			+ c_3\sum_{(r, b) \notin \EuScript{I}} \hat{\transp}(r, b) [\cutlen_1(r, b) + \cutlen_2(r, b)]
			+ c_3\sum_{(r, b) \in \EuScript{I}} \hat{\transp}(r, b) [\cutlen_1(r, b) + \cutlen_2(r, b)]
			+ \Gamma
			\\
		&\leq \tfcost(\hat{\transp}) 
			+ c_3\sum_{(r, b) \in \reds \times \blues} \hat{\transp}(r, b) [\cutlen_1(r, b) + \cutlen_2(r, b)]
			+ \Gamma,
	\end{align*}
	by choosing $c_3 \geq (2c_2 + 20\sqrt{2}) \geq 100\sqrt{2}$.

	To complete the proof, it suffices to show that $\Gamma \leq 0$.
	Observe that the symmetric difference between individual units of transportation in $\hat{\transp}$ and 
	$\tilde{\transp}$ can be decomposed into a set of alternating 4-cycles, 
	where both edges of $\hat{\transp}$ and at most one edge of 
	$\tilde{\transp}$ are non-local 
	(c.f. construction in Lemma~\ref{lemma:grid_swap}).
	Let the points of the cycle be $a, b, c, d$, where
	$(b, c), (d, a) \in \tilde{\transp}$ and 
	$(a, b), (c, d) \in \hat{\transp}$.
	Without loss of generality, suppose $(b, c)$ is the 
	non-local $\tilde{\transp}$ edge.
	When $(a, b, c, d)$ is passed to the external subproblem,
	local edges are contracted -- thus the corresponding cycle
	on the merged points $a', b', c', d'$ is a triangle 
	$(a', b', c')$, and $a' = d'$.

	Suppose now that a grid of cell length $\Delta^*$ cuts 
	$(b', c')$ during the secondary recursion,
	i.e. $\cutlen_2(b, c) = \Delta^*$.
	Since $(a', b', c')$ is a triangle, one of $(a', b')$ or $(c', a')$
	is also cut by the same grid.
	This implies that $\cutlen_2(b, c) \leq \cutlen_2(a, b) + \cutlen_2(c, d)$.
	Summing over all the cycles and using Lemma~\ref{lemma:grid_swap}\ref{item:grid_swap3},
	\begin{equation*}
		\sum_{(r, b) \notin \EuScript{I}} \tilde{\transp}(r, b) \cutlen_2(r, b)
			\leq \sum_{(r, b) \notin \EuScript{I}} \hat{\transp}(r, b) \cutlen_2(r, b),
	\end{equation*}
	and $\Gamma \leq 0$, as desired.
	Then,
	\begin{equation*}
		\tfcost(\transp) \leq \tfcost(\hat{\transp}) 
			+ c_3\sum_{(r, b) \in \reds \times \blues} \hat{\transp}(r, b) [\cutlen_1(r, b) + \cutlen_2(r, b)]
	\end{equation*}
	which completes the proof.
\end{proof}

\subsection{An efficient implementation}
\label{subsection:grid_efficient}

We now explain how the various steps of the algorithm are implemented 
to run in $O(n^{1+\eps})$ time.
There are three main steps in the algorithm:
\begin{enumerate}
\item \label{item:grid_efficient1}
	partitioning a recursive subproblem $\rinst = \rtinst{}{}$ into 
	internal subproblems and an external subproblem;
\item \label{item:grid_efficient2}
	solving subproblems recursively;
\item \label{item:grid_efficient3}
	recovering the transportation map $\transp$ of $\rinst$ 
	from the internal and external solutions $\transp_\pi$ 
	($\pi \in \Pi$) and $\transp_\Box$.
\end{enumerate}
A recursive subproblem partitions its points into internal subproblems,
but generates an additional set of points (at cell centers) for its
external subproblem; let us call these \emph{external points}.
We bound the number of external points generated in the next lemma.
\begin{lemma}
\label{lemma:grid_efficient_external_count}
	The total number of external points over all recursive subproblems is $O(n)$.
\end{lemma}
\begin{proof}
	We first prove the lemma in the case of bounded spread
	(external subproblems are solved exactly).
	Each external point $p$ is the center of a cell $\pi$ that is either
	bichromatic (contains points of both colors) 
	or monochromatic (only one color).
	If $\pi$ is monochromatic, 
	we charge $p$ to any of the $\reds_{ex} \cup \blues_{ex} \subseteq \reds \cup \blues$ 
	points in $\pi$ --- by construction, these points are omitted from the internal 
	subproblems and do not appear in any later recursive subproblems.
	If $\pi$ is bichromatic, 
	we charge $p$ to the internal subproblem created on $\pi$.
	Thus, the number of external points is at most
	$(n + [\text{\# internal subproblems}])$.

	When spread is bounded, 
	all internal subproblems are in the primary recursion.
	Consider a problem instance $\inst$ with $m$ points, and let $I(m)$
	denote the maximum number of internal subproblems 
	that can be recursively generated by it.
	We will inductively prove that $I(m) \leq 2m-1$. 
	If the subdividing grid for $\inst$ generates a single 
	internal subproblem, then that subproblem contains at most 
	$m-1$ points. By the inductive hypothesis, this subproblem 
	generates at most $2(m-1)-1$ points. Adding $\inst$ itself
	gives a total of $2m-2 < 2m-1$ internal subproblems.
	On the other hand, if the subdividing grid for $\inst$
	generates $k > 1$ internal subproblems with 
	$m_1, m_2, \ldots, m_k$ points respectively, then
	$\sum_{i=1}^k m_i \leq m$ since each point can be 
	inherited into at most one internal subproblem. 
	Using the inductive hypothesis on these subproblems
	gives a bound of $\sum_{i=1}^k (2m_i-1) \leq 2m-k$
	on the number of internal subproblems generated recursively. 
	Adding $\inst$ to this sum
	gives a total of $2m-k+1 \leq 2m-1$ internal subproblems.
	%
	%Internal subproblems partition the $m$ points between them,
	%so if there are at least 2 internal subproblems in the subdivision,
	%the number of internal subproblems is at most $n$.
	%However, this may not always be the case --- 
	%the subdivision may generate only 1 internal subproblem.
	%If there is only 1 internal subproblem,
	%then the subdivision must have created at least one nonempty
	%grid cell with no internal subproblem, i.e. a monochromatic cell.
	%Thus, the sole internal subproblem must have fewer than $m$ points,
	%and the number of internal subproblems generated this way is at most $n$.
	It follows that for an input instance of bounded spread with
	$n$ points, the number of internal subproblems is at most $2n$,
	and the number of external points is at most $3n$.

	In the general case, we can apply a similar analysis
	for the subproblems of the primary recursion,
	although there are now external points and internal subproblems 
	in the secondary recursion.
	However, we can apply the bounded spread analysis to each 
	external subproblem that starts a secondary recursion.
	We have
	\begin{align*}
		[\text{\# secondary internal subproblems}] 
			&\leq 2 \cdot [\text{\# primary external points}] \\
			&\leq 2 \cdot (n + [\text{\# primary internal subproblems}]) \\
			&\leq 2(n + (2n))
			\leq 6n.
	\end{align*}
	Hence, the number of secondary external points is at most $(n + 6n) = 7n$,
	and the total number of external points is at most $3n + 7n = 10n = O(n)$.
\end{proof}

A base subproblem of size $n_i \leq n^{\eps/4}$ is solved in 
$O(n_i^3 \log n_i)$ time using Orlin's algorithm.
We distribute this on the $n_i$ points by charging 
$O(n_i^2 \log n_i) = O(n^{\eps/2} \log n)$ to each point.
Note every point in $\reds \cup \blues$, as well as every external point, 
belongs to at most one base subproblem.
Since, by Lemma~\ref{lemma:grid_efficient_external_count}, the
number of external points is $O(n)$, it follows that the total time spent
solving base subproblems is 
$O(n) \cdot O(n^{\eps/2} \log n) = O(n^{1+\eps})$.

The time spent in recovering the transportation map $\transp$ for $\rinst$ 
from its internal and external subproblems is proportional to the number of
external points in $\rinst$.
Hence, by Lemma~\ref{lemma:grid_efficient_external_count}, 
step~\ref{item:grid_efficient3} takes $O(n)$ time.

Finally, the implementation of step \ref{item:grid_efficient1} depends on whether 
the instance $\inst = \tinst{}{}$ has bounded spread.

\mparagraph{The bounded spread case}
%\subparagraph*{The bounded spread case.}
In this case, step~\ref{item:grid_efficient1} is implemented naively.
We choose a random shift, distribute the points of $\rreds \cup \rblues$
among the grid cells in $O(m \log m)$ time, where $|\rreds \cup \rblues| = m$,
and check in additional $O(m)$ time whether the shift is safe.
So step~\ref{item:grid_efficient1} can be implemented in $O(m \log m)$ expected time.
We charge $O(\log m)$ time to each point of $\rreds \cup \rblues$.
Since the spread is $n^{O(1)}$, the depth of recursion is $O(1/\eps)$.
Therefore, each input point is charged $O(\frac{1}{\eps}\log n)$ units of time, 
implying that steps \ref{item:grid_efficient1} over all subproblems 
take $O(n\log n)$ expected time total (recall that $\eps$ is a constant). 
As the size of the external subproblem at $\rreds \cup \rblues$ is $O(m^{2 \delta})$, 
and $\delta = 1/6$, the time for solving it exactly is $O(m \log m)$. 
Again, over all levels of recursion, the total time spent solving external subproblems is~$O(n\log n)$.
Putting everything together and applying 
Corollary~\ref{corollary:grid_approx_bs}, we obtain the following.
\begin{theorem}
\label{theorem:grid_time_bs}
	Let $\inst$ be an instance of the transportation problem in $\reals^d$, where $d$ is a constant.
	Let $\inst$ have size $n$ and bounded spread, and let $\eps > 0$ be a constant.
	A transportation map of $\inst$ can be computed in $O(n^{1+\eps})$ 
	expected time whose expected cost is $O(\log (1/\eps))\tfcost(\transp^*)$,
	where $\transp^*$ is an optimal transport of $\inst$.
\end{theorem}

\mparagraph{The general case}
%\subparagraph*{The general case.}
Let $\rinst = \rtinst{}{}$ be a recursive subproblem, and
$\Box$ the smallest square containing $\rreds \cup \rblues$. 
As defined earlier, let $\rinst_\Box = \rtinst{}{_\Box}$ be the external 
subproblem generated by $\rinst$ using the points of $\rreds_{ex}, \rblues_{ex}$. 
Since $\rinst_\Box$ has bounded spread and is solved recursively,
the running time to solve $\inst_\Box$ is 
$O(|\rreds_\Box \cup \rblues_\Box|^{1+\eps}) = O(|\rreds_{ex} \cup \rblues_{ex}|^{1+\eps})$.
Summing over all recursive problems, 
by Lemma~\ref{lemma:grid_efficient_external_count}, the total time spent in 
solving all external subproblems is $O(n^{1+\eps})$.

Step~\ref{item:grid_efficient1} is more challenging in this case because 
the depth of recursion can be as large as $\Omega(n)$.
We can neither spend linear time at a recursive subproblem, nor can we afford 
to visit each cell of the randomly shifted grid $\rsg$ explicitly 
to compute the set $\Pi$ of non-empty cells.
To avoid checking all cells of $\rsg$ explicitly, we (implicitly) construct a 
quad tree $\rsqt$ on $\rsg$ --- 
i.e. leaves of $\rsqt$ are cells of $\rsg$ and the root of $\rsqt$ is $\Box$.
The depth of $\rsqt$ is $O(\log m)$.
The role of $\rsqt$ will be to guide the search for non-empty cells 
of $\rsg$.

To avoid spending $\Omega(m)$ time in step~\ref{item:grid_efficient1}, 
we do not maintain the set $\rreds \cup \rblues$ explicitly.
We build a 2D dynamic orthogonal range searching data structure that 
maintains a set $X \subset \reals^2$ of weighted points,
and supports the following operations:
\begin{itemize}
\item $\Fn{Wt}(\gsprect)$: Given a rectangle $\gsprect$, return $w(X \cap \gsprect)$.
\item $\Fn{Report}(\gsprect, \Delta)$: Report a maximal subset $Y$ of 
	$X \cap \gsprect$ such that $w(Y) \leq \Delta$.
\item $\Fn{Empty}(\gsprect)$: Return $\textsc{Yes}$ if 
	$X \cap \gsprect = \emptyset$ and $\textsc{No}$ otherwise.
\item $\Fn{Delete}(p)$: Delete $p$ from $X$.
\item $\Fn{ReduceWt}(p, \Delta)$: Update $w(p) := w(p) - \Delta$, 
	assuming $w(p) \geq \Delta$.
\end{itemize}
Using a range-tree based data structure, each operation except for $\Fn{Report}$ can be 
performed in $O(\log^2 n)$ time~\cite{agarwal1999geometric}.
$\Fn{Report}$ requires $O(\log^2 n + k)$ time, 
where $k$ is the number of reported points.

We maintain two copies $\mathbb{D}_\reds, \mathbb{D}_\blues$ 
of this data structure.
The first one is initialized with $\reds$ and the supplies,
and the second with $\blues$ and the demands.
We use $\EuScript{\reds}$ and $\EuScript{\blues}$ to denote 
the current sets in these data structures.

With each recursive subproblem $\rinst = \rtinst{}{}$ 
we associate a bounding rectangle $\gsprect$ that contains 
$\rreds \cup \rblues$.
For the root problem, $\gsprect$ is the smallest square containing 
$\reds \cup \blues$; for others it is defined recursively.
We maintain the invariant that when the subproblem 
$\rinst = \rtinst{}{}$ is being processed,
\begin{itemize}
\item $\EuScript{\reds} \cap \gsprect = \rreds$ and for any 
	$r \in \EuScript{\reds} \cap \gsprect, w(r) = \rsupply(r)$,
\item $\EuScript{\blues} \cap \gsprect = \rblues$ and for any
	$b \in \EuScript{\blues} \cap \gsprect, w(b) = \rsupply(b)$.
\end{itemize}

We first compute $\Pi$, the set of non-empty cells of $\rsg$, using $\rsqt$ 
and the data structures $\mathbb{D}_\reds, \mathbb{D}_\blues$.
We visit $\rsqt$ in a top-down manner.
Suppose we are at a node $v \in \rsqt$, 
and let $\Box_v$ be the square associated with $v$.
We call $\Fn{Empty}(\gsprect \cap \Box_v)$ on both $\mathbb{D}_\reds$ and 
$\mathbb{D}_\blues$ to check whether 
$(\rreds \cup \rblues) \cap \Box_v = \emptyset$.
If \textsc{Yes}, we ignore the subtree rooted at $v$.
If \textsc{No} and $v$ is a leaf of $\rsqt$, i.e., 
$\Box_v$ is a nonempty cell of $\rsg$, we add $\Box_v$ to $\Pi$.
If $v$ is an internal node and 
$(\rreds \cup \rblues) \cap \Box_v \neq \emptyset$,
we recursively search the children of $v$ in $\rsqt$. 
Since the depth of $\rsqt$ is $O(\log n)$, the above procedure visits 
$O(|\Pi| \log n)$ nodes of $\rsqt$.
The total time spent in computing $\Pi$ is thus $O(|\Pi| \log^3 n)$.

For each cell $\pi \in \Pi$, we can compute the total demands of
$\rsupply (\rreds \cap \pi)$ and $\rsupply (\rblues \cap \pi)$ --- 
and thus $\excess_\pi$ --- using the $\Fn{Wt}(\gsprect \cap \pi)$ 
operations on $\mathbb{D}_\reds, \mathbb{D}_\blues$.
Without loss of generality, 
assume $\rsupply (\rreds \cap \pi) > \rsupply (\rblues \cap \pi)$.
Using $\Fn{Report}(\gsprect \cap \pi, \excess)$, we report a 
maximal subset of points of $\rreds \cap \pi$ whose total weight 
is at most $\excess$.
We then delete each of these points (by $\Fn{Delete}$)
and reduce the weight (by $\Fn{Reduce}$) of one additional point in 
$\rreds \cap \pi$ or $\rblues \cap \pi$ if needed.
Let $\tinst{}{_\pi}$ be the recursive (internal) subproblem generated 
for $\pi$ with $\gsprect_\pi = \gsprect \cap \pi$ as the associated rectangle.
Then the above update operation ensures that 
$\EuScript{\reds} \cap \gsprect_\pi = \rreds_\pi$, 
$\EuScript{\blues} \cap \gsprect_\pi = \rblues_\pi$, 
and their weights are consistent with $\rsupply_\pi$.

$\mathbb{D}_\reds$ and $\mathbb{D}_\blues$ can also be used to test whether 
the random shift is safe:
For each $\pi \in \Pi$, we check whether the moat of any point in 
$(\rreds \cup \rblues) \cap \pi$ intersects an edge of $\pi$.
This is equivalent to checking whether 
$\qtsquare_\pi \cap (\rreds \cup \rblues) = \emptyset$,
where $\qtsquare_\pi$ is the set of points that are within distance $\ell/m^3$
from the boundary of $\pi$.
This test can be done in $O(\log^2 n)$ time using the $\Fn{Empty}$ procedure.
The total expected time spent in generating internal subproblems $\rinst_\pi$ 
and external subproblems $\rinst_\Box$ of $\inst$ is 
$O(|\Pi|\log^3 n + m_\Box \log^2 n)$,
where $m_\Box$ is the total number of external points.

By Lemma~\ref{lemma:grid_efficient_external_count}, 
the total number of nonempty cells over all subproblems is $O(n)$, 
and the number of external points is $O(n)$.
Thus, the expected time spent in step~\ref{item:grid_efficient1} overall is $O(n\log^3 n)$.
Putting everything together, we obtain the main result of this section.

\begin{theorem}
\label{theorem:grid_time_gen}
	Let $\inst$ be an instance of the transportation problem in $\reals^d$, where $d$ is a constant.
	Let $\inst$ have size $n$, and let $\eps > 0$ be a constant.
	A transportation map of $\inst$ can be computed in $O(n^{1+\eps})$ 
	expected time whose expected cost is $O(\log^2 (1/\eps))\tfcost(\transp^*)$,
	where $\transp^*$ is an optimal transport of $\inst$.
\end{theorem}

\subsection{Extension to bounded doubling dimension}

The \emph{doubling dimension} of a metric space is the minimum $D$
such that every radius $2r$ ball can be covered by at most $2^D$ 
balls of radius $r$. 
When $D$ is a constant, we say the metric space is \emph{doubling}.
We now describe the changes to the algorithm when we are 
given oracle access to a doubling metric $\dist{\cdot}{\cdot}$.
This result has a logarithmic dependence on the spread $\Phi$,
and therefore gives the same asymptotic running times only when
$\Phi = n^{O(1)}$.
The dependence on spread exists because we lack the 
``efficient datastructures'' of Section~\ref{subsection:grid_efficient}
for doubling metrics.

A regular grid is not well-defined in this setting, 
but we can replace the partitioning from a randomly shifted grid 
with a partitioning generated by the \emph{randomized split-tree} of Talwar~\cite{talwar2004bypassing}.
In the rest of this section, we will describe that this new partitioning
has similar bounds on ``cell'' count, ``cell'' diameter, and point separation probability.
Afterwards, we state the split-tree construction in-depth,
which allows us to define and efficiently implement a notion of ``safety'' for the new subdivision.

\subsubsection{Replacing subdividing grid}

A split-tree is a hierarchical decomposition of the metric, 
and Talwar's construction has properties reminiscent of a 
randomly shifted quadtree in Euclidean space.
Each node of the split-tree corresponds to a cluster of points,
whose children are a partitioning of itself, and so forth.
The root contains all points, and each leaf contains a single point.

Formally, we number the \emph{levels} of the tree,
with the root at level $L = O(\log \ell)$, 
where $\ell$ is the diameter of the input point set.
The children of level $i$ nodes at level $(i-1)$.
Similar to the quadtree, the hierarchy contracts cluster diameters by $1/2$ in successive levels.
The root of the tree is a cluster of all points,
whereas leaves are singleton clusters for each point.
Talwar's construction (which we describe next) has the following properties:
\begin{enumerate}
\item \label{item:grids_doubling_prop1}
	A level $i$ cluster has diameter at most $2^{i+1}$.
\item \label{item:grids_doubling_prop2}
	A level $i$ cluster is comprised of at most 
	$2^{O(D)}$ level $(i+1)$ clusters.
\item \label{item:grids_doubling_prop3}
	For any two points $r, b$, 
	the probability that $r$ and $b$ lie in different
	level $i$ clusters is at most 
	$O(D)\cdot \dist{r}{b}/2^{i}$.
\end{enumerate}
Recall that we use $m$ to denote the number of points in a recursive instance on $\rreds \cup \rblues$.
For a point partitioning analagous to an $[m^\delta]^D$ grid, 
it suffices to return the $(L - \delta \log_2 m)$-level of the randomized split-tree.
Denote the partitioning induced by this level as $\rsg_D$.

Even in the Euclidean case, $\delta$ is a constant parameter that depended on the dimension.
For $\rsg_D$, let the fanout bound in (\ref{item:grids_doubling_prop2}) be $2^{c'D}$, for some constant $c' > 0$.
For the doubling metric, we choose $\delta = (3c'D)^{-1}$, to achieve $O(m^{c'D \delta}) = O(m^{1/3})$.
With this choice of $\delta$, the external subproblem at $\rsg_D$ can be solved in $O(m)$ time.
\begin{lemma}
\label{lemma:grid_doubling_rsgd}
	Let $\rreds \cup \rblues$ be a set of $m$ points in a doubling metric, and $\ell$ be their diameter.
	For a parameter $\delta > 0$, let $\rsg_D$ be the partitioning of $\rreds \cup \rblues$ given by the 
	$(L - \delta \log_2 m)$-level of the randomized split tree.
	Then for $\Delta = \ell/m^\delta$,
	\begin{enumerate}
	\item $\rsg_D$ consists of $O(m^{c'D\delta}) = O(m^{1/3})$ clusters.
	\item Each cluster in $\rsg_D$ has diameter $2^L/m^\delta = c \Delta$, 
		for some constant $c > 0$.
	\item For any $(r, b) \in \rreds \times \rblues$, 
		the probability that $r$ and $b$ lie in different
		clusters of $\rsg_D$ is at most $O(\dist{r}{b}/\Delta)$.
	\end{enumerate}
\end{lemma}
\begin{proof}
	Use level $(L - \delta \log_2 m)$ in the randomized split-tree properties.
\end{proof}
$\rsg_D$ can be made ``safe'' in a similar way as the Euclidean case.
Formally, we say $\rsg_D$ is \emph{safe} if, 
for all $(r, b) \in \rreds \times \rblues$ where $\dist{r}{b} \leq c\Delta/m^3$,
$r, b$ lie in the same cluster of $\rsg_D$.
Here, $c$ is the constant that appears in the diameter bound for a $\rsg_D$ cluster.
To describe the changes and demonstrate that they can be implemented efficiently, 
we first detail Talwar's construction algorithm.

\subsubsection{Talwar's construction}

Suppose a set $S$ of $m$ points in the doubling metric.
Talwar's construction builds the split-tree from a \emph{net hierarchy}.
An \emph{$r$-net} of a point set $S$ is a subset $N \subset S$ which satisfies:
\begin{itemize}
\item \emph{Covering}: for all $p \in S$, there exists some $q \in N$ with $\dist{p}{q} \leq r$.
\item \emph{Packing}: for all $p, q \in N$, $\dist{p}{q} \geq r$.
\end{itemize}
Now, suppose that we build a nested series of nets $S = N_0, N_1, \ldots, N_L$
such that $N_i$ is a $2^{i-2}$-net of $N_{i-1}$.
Note that $N_L$ consists of a single point.
Such a net hierarchy can be constructed in $2^{O(D)}m\log m$ expected time using 
the algorithm of Har-Peled and Mendel~\cite{har2005fast}.

Given a net hierarchy on $S$, Talwar's algorithm is as follows.
We will use $\EuScript{C}_i$ to denote the set of clusters at level $i$.
\begin{enumerate}
\item Pick $\rho$ uniformly at random from $[1/2, 1)$,
	a random permutation $\pi$ of $\reds \cup \blues$,
	and initialize $\EuScript{C}_L = \{\reds \cup \blues\}$.
\item Repeat for each level $i$, from $L$ to 0.
	Let $r_i = 2^i \rho$.
	For each cluster $C \in \EuScript{C}_{i+1}$, we generate its children clusters
	(into $\EuScript{C}_i$) by iterating through the net points $N_i$ in order of $\pi$.
	The points of $N_i$ will form the cluster centers at $\EuScript{C}_i$.
	For each $p \in N_i \cap C$ in order of $\pi$, create a new cluster from the yet-unassigned points of $C$ in the $r_i$-ball about $p$.
	Formally, we create the child cluster $C'(p)$ defined as
	\begin{equation*}
		C'(p) := \{q \in C \mid \text{$\dist{p}{q} \leq r_i$, and $\forall q' \in N_i$ where $\pi(q') < \pi(q)$, $\dist{p}{q'} > r_i$} \}.
	\end{equation*}
\end{enumerate}
Talwar's construction can be implemented in $O(2^{O(D)} m)$ time per level.
The doubling property ensures that $|N_i \cap C| = 2^{O(D)}$,
so we can compare each point naively against all the points of $N_i \cap C$.
The total time to construct the split-tree is $O(m\log m) + O(m L) = O(m\log m + m\log\Phi)$ time.
To find $\rsg_D$, we stop the construction after $\delta \log_2 m$ levels,
in time $O(m\log m + m\delta \log m) = O(m \log m)$ (recall $\delta = 1/6$).

\subsubsection{$\rsg_D$ safety}

To determine if $\rsg_D$ is safe, it suffices to check a local condition 
at every split-tree cluster from level $L$ to $(L - \delta\log_2 m)$.
First, we define two predicates for a separation property between a cluster and points outside it. 
Let $C$ be a split-tree cluster with center $q$, construction radius $r$, and points $S \supseteq C$.
\begin{itemize}
\item $P(C, S)$: for all $p \in S \setminus C$ and $p' \in C$, $\dist{p}{p'} > c\Delta/m^3$.
\item $Q(C, S)$: for all $p \in S \setminus C$, $\dist{p}{q} > r + c\Delta/m^3$.
\end{itemize}
By definition, $\rsg_D$ is safe if $P(C, \rreds \cup \rblues)$ holds for all $C \in \rsg_D$.

\begin{lemma}
\label{lemma:grid_doubling_localsafe}
	For a split-tree cluster $C$, let $\operatorname{par}(C)$ be its parent cluster.
	If $Q(C, \operatorname{par}(C))$ holds for all $C$ in levels $L$ to $(L - \delta\log_2 m)$,
	then $P(C', \rreds \cup \rblues)$ holds for all $C' \in \rsg_D$.
\end{lemma}
\begin{proof}
	We prove, by induction on levels, that $P(C, \rreds \cup \rblues)$ for clusters at every level in the range.
	The root cluster $C_{root}$ (at level $L$) excludes no vertices, so $P(C_{root}, \rreds \cup \rblues)$ holds vacuously.

	Suppose the level $i+1$ clusters have $P(C, \rreds \cup \rblues)$.
	Consider a level $i$ cluster $C$ with center $q$ and radius $r$.
	For any $p \not\in C$, either
	\begin{enumerate}
	\item $p \not\in \operatorname{par}(C)$ as well, or
	\item $p \in \operatorname{par}(C)$, and was first separated from the points of $C$ at level $i$.
	\end{enumerate}
	If $p \not\in \operatorname{par}(C)$, then for any $p' \in C$ we have 
	$P(\operatorname{par}(C), \rreds \cup \rblues) \implies \dist{p}{p'} > c\Delta/m^3$,
	as $C$ is a subset of $\operatorname{par}(C)$.
	Otherwise $p \in \operatorname{par}(C)$, and we utilize the assumption that $Q(C, \operatorname{par}(C))$ holds.
	For any $p' \in C$, triangle inequality implies that
	\begin{equation*}
	\begin{aligned}
		\dist{p}{p'} &\geq \dist{p}{q} - \dist{p'}{q} \\
			&> r + c\Delta/m^3 - r \\
			&= c\Delta/m^3
	\end{aligned}
	\end{equation*}
	Since the separation property holds for all points of $\overline{C}$, 
	we have established $P(C, \rreds \cup \rblues)$ for all level $i$ clusters.
	The lemma follows by induction down to the $(L - \delta\log_2 m)$-level, $\rsg_D$.
\end{proof}

This suggests an efficient procedure for checking if $\rsg_D$ is safe.
While constructing $\rsg_D$ through the randomized split-tree,
we verify $Q(C, \operatorname{par}(C))$ point-by-point as we add them to clusters.
Since we already do these naive distance comparisons to cluster the points,
the asymptotic running time of the construction stays the same.
Like in the Euclidean case, we restart the construction with new random parameters 
if $Q(C, \operatorname{par}(C))$ fails on any cluster.
We finish this section by showing that a random $\rsg_D$ is safe with high probability,

\begin{lemma}
	The randomized split-tree construction produces a safe $\rsg_D$ with probability $\Omega(1 - 1/m^{5/3})$.
\end{lemma}
\begin{proof}
	We prove that $\rsg_D$ satisfies $Q(C, \operatorname{par}(C))$ at every cluster $C$ 
	with probability $\Omega(1 - 1/m^{5/3})$.
	By Lemma~\ref{lemma:grid_doubling_localsafe}, $\rsg_D$ meeting this condition is safe.

	Consider the random radius parameter, $\rho \in [1/2, 1)$, fixed at the start of the construction.
	At level $i$, the split-tree construction choosese cluster radii $r_i = 2^i \rho$.
	For a cluster $C$ with center $q$, $Q(C, \operatorname{par}(C))$ fails if the
	circle of radius $r_i$ centered at $q$ intersects the $(c\Delta/m^3)$-ball
	of any point outside of it.
	The radial measure of these balls is no more than
	\begin{equation*}
		m \cdot \frac{2c\Delta}{m^3} \leq \frac{4 r_i}{m^2} \leq \frac{4 \cdot 2^i}{m^2}
	\end{equation*}
	as $\rho c\Delta$ is the smallest cluster radius and $\rho \geq 1/2$.
	Then, the probability that $rho$ causes $C$ to fail $Q(C, \operatorname{par}(C))$ is no more than
	\begin{equation*}
		\frac{4 \cdot 2^i / m^2}{2^i - 2^{i-1}} = \frac{2^{i+2}}{2^{i-1} m^2} = \frac{8}{m^2}.
	\end{equation*}
	By Lemma~\ref{lemma:grid_doubling_rsgd} the number of $\rsg_D$ clusters is $O(m^{1/3})$.
	Since the split-tree has constant degree, the total number of clusters is also $O(m^{1/3})$.
	By union bound, the probability that any $C$ fails is no more than
	\begin{equation*}
		O(m^{1/3}) \cdot \frac{8}{m^2} = O(m^{-5/3})
	\end{equation*}
	and the probability that no $C$ fails is $\Omega(1 - 1/m^{5/3})$.
\end{proof}

% merged into wspd.tex
%\section{Minimum Cost Flow Preliminaries}
%\label{section:mcf}
%\input{mcf}

\section{A $(1+\epsilon)$-Approximate Algorithm}
\label{section:wspd}
In this section, we describe a $(1 + \eps)$-approximation algorithm based 
on a reduction to minimum cost flow.
We begin by defining the minimum cost flow problem.
Our use of minimum cost flow in this section is relatively black box,
so we postpone some algorithmic definitions until Section~\ref{section:orlin},
like the residual graph and the dual problem.

\subsection{Minimum cost flow definitions}
\label{subsection:mcf_def}

Let $G_0 = (V, E)$ be a directed graph with $n$ vertices and $m$ arcs,
with \emph{costs} $c: E \to \reals$, \emph{capacities} $u: E \to \reals_+$,
and signed supplies/demands $\fsupply: V \to \ints$ where 
$\sum_{v \in V} \fsupply(v) = 0$.
We call $\fsupply(\cdot)$ \emph{flow-supply} (when $\fsupply(v) > 0$) or 
\emph{flow-demand} (when $\fsupply(v) < 0$) to distinguish from the 
nonnegative supply-demand function $\tsupply(\cdot)$ used for transportation.
We use $C = \max_{e \in E} c(e)$ and $U = \max_{v \in V} |\fsupply(v)|$;
it is assumed that the costs are scaled such that the minimum arc cost is 1.

A \emph{pseudoflow} is an arc function $f: E \to \reals_+$ satisfying the
\emph{capacity constraints}
\begin{equation*}
	f(v, w) \leq u(v, w) \qquad \forall (v, w) \in E.
\end{equation*}
If an arc $(v, w) \in E$ has $f(v, w) > 0$, we say it is \emph{active}; otherwise it 
is said to be \emph{idle}. 
We call the active arcs of a pseudoflow its \emph{support}.
The cost of a pseudoflow is
\begin{equation*}
	\cost(f) := \sum_{e \in E} c(v, w) f(v, w).
\end{equation*}
With respect to a pseudoflow $f$, the \emph{imbalance} of a vertex $v \in V$ is
\begin{equation*}
	e_f(v) := \fsupply(v) + \sum_{(w, v) \in E} f(w, v) - \sum_{(v, w) \in E} f(v, w).
\end{equation*}
We refer to vertices with $e_f(v) > 0$ as \emph{excess vertices}, and those 
with $e_f(v) < 0$ as \emph{deficit vertices}.

If imbalance is zero on all vertices, $f$ is a \emph{flow}.
The \emph{minimum cost flow} problem (MCF) is to find a flow $f^*$ of minimum 
cost.
Transportation can be formulated as a special case of MCF on the complete 
bipartite graph, using $V = \reds \cup \blues$, $E = \reds \times \blues$ 
directed $\reds$ to $\blues$, $c(r, b) = \|r - b\|$, 
setting flow-supplies $\fsupply(r) = \tsupply(r)$ for $r \in \reds$
and flow-demands $\fsupply(b) = -\tsupply(b)$ for $b \in \blues$.
The transportation map can be recovered from a flow solution as 
$\tau(a, b) = f(a, b) \forall a \in A, b \in B$.

\subsection{Algorithm outline}

As in Section~\ref{section:grids}, we hierarchically cluster points,
but this time for the purpose of approximately representing 
all $\Theta(n^2)$ pairwise distances between $\reds$ and $\blues$ compactly.
At a high level, our algorithm is as follows:

\begin{enumerate}
\item \label{item:wspd_1}
	Compute a hierarchical clustering of $\reds \cup \blues$
	using a quadtree with a single input point in each leaf.

\item \label{item:wspd_2}
	Construct a sparse directed acyclic graph 
	$\sparseG = (\sparseV, \sparseE)$
	over the clusters with $\reds$ and $\blues$ inheriting their supplies 
	and demands respectively from the transportation instance.
	$\sparseG$ has the following property:
	for every pair $(r, b) \in \reds \times \blues$, there is a unique path
	in $\sparseG$ from $r$ to $b$, with cost roughly $\dist{r}{b}$.
	That is, the per-unit cost of flow from $r$ to $b$ in $\sparseG$
	approximates the per-unit cost of transporting from $r$ to $b$ in the 
	original metric.

\item \label{item:wspd_3}
	Compute an optimal flow $\flow^*$ in $\sparseG$ using a minimum 
	cost flow algorithm.
	To achieve our advertized running time, we use the algorithm by 
	Lee and Sidford~\cite{LS-flow}.

\item \label{item:wspd_4}
	Recover a transportation map $\transp$ from $\flow^*$.
\end{enumerate}

The hierarchical structure is the foundation for the compact representation 
-- a well-separated pair decomposition~\cite{callahan1995decomposition} --
and further enables a near-linear time procedure for 
recovering the transportation map (step \ref{item:wspd_4}).
This fast recovery step is what distinguishes our algorithm from the similar reduction in
Cabello~\etal~\cite{cabello2008matching}, which uses \emph{geometric
spanners}~\cite{callahan1993faster,arya1994randomized,arya1995euclidean} instead.
As black boxes, geometric spanners seem to be insufficient 
for efficient recovery of the transportation map.

%TODO we need to associate leaf nodes with the single point they contain
% they may not be leaves, though, if the other color isn't "leafing" there.

\subsection{Construction of the graph}

Let $\square$ be the smallest orthogonal square containing $\reds \cup \blues$.
We construct a \emph{compressed quad tree} $\qtree$ on $\reds \cup \blues$ 
with $\square$ as the square associated with the root of $\qtree$.
A compressed quadtree prunes certain interior nodes of a standard quadtree,
guaranteeing that $\qtree$ has $O(n)$ nodes.
We can construct a compressed quadtree in $O(n \log n)$ time,
see e.g. \cite{har2011geometric}.

Each node $v$ of $\qtree$ is associated with a square $\qtsquare_v$.
For each node $v \in \qtree$, let $\reds_v = \reds \cap \qtsquare_v$
and $\blues_v = \blues \cap \qtsquare_v$.
If $\reds_v$ (resp. $\blues_v$) is a singleton set, 
we call $v$ a \emph{red leaf} (resp. \emph{blue leaf}),
and denote its associated point of $\reds$ by $\ell_\reds(v)$ (resp. $\ell_\blues(v)$, a $\blues$ point).
The sets $\reds_v, \blues_v$ form a hierarchical clustering of 
$\reds \cup \blues$.

To construct $\sparseG = (\sparseV, \sparseE)$,
we make two copies of $\qtree$: 
the \emph{up-tree} $\uptree = (\uptreeV, \uptreeE)$
and \emph{down-tree} $\downtree = (\downtreeV, \downtreeE)$.
We orient the arcs of $\uptreeE$ upward --- from a node to its parent,
and orient the arcs of $\downtreeE$ downward --- from a node to its child.
We delete blue points from $\uptree$ and red points from $\downtree$,
thus $\uptree$ contains only $\reds$, and $\downtree$ contains only $\blues$.
We set $\sparseV = \uptreeV \cup \downtreeV$ and 
$\sparseE = \uptreeE \cup \downtreeE \cup \crossE$ where 
$\crossE \subseteq \uptreeV \times \downtreeV$ is a set of \emph{cross edges}
connecting $\uptree$ to $\downtree$ that we define next.
See Figure~\ref{figure:wspd_construction}.

\begin{figure}
\centering
\subcaptionbox{\centering Quadtree on $\reds \cup \blues$.}{\includegraphics[width=0.25\linewidth,page=1]{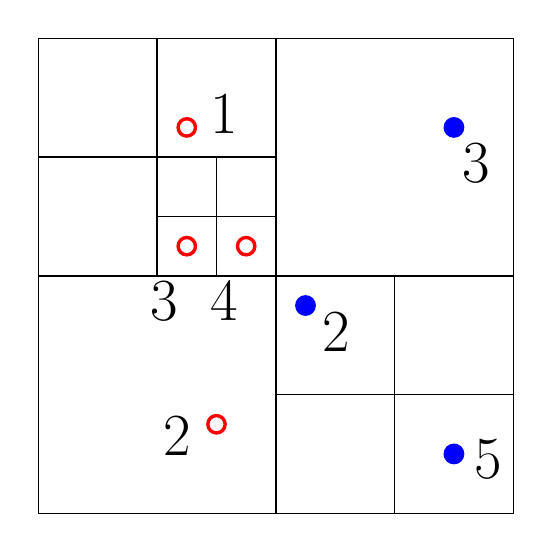}}
\quad
\subcaptionbox{\centering Composition of $\sparseG$.}{\includegraphics[width=0.3\linewidth,page=3]{wspd_example.pdf}}
\quad
\subcaptionbox{\centering Decomposition of flow.}{\includegraphics[width=0.3\linewidth,page=4]{wspd_example.pdf}}
\caption{Construction of $\sparseG$ and recovering the transportation.}
\label{figure:wspd_construction}
\end{figure}

\mparagraph{Well-separated pair decomposition}
Originally proposed by Callahan and Kosaraju~\cite{callahan1995decomposition},
the notion of a \emph{well-separated pair decomposition} (WSPD) of a 
point set $S$ is widely used to represent the pairwise distances of 
$S$ approximately in a compact manner.
The well-separated pair decomposition of a point set $P$ is a set of
subset pairs $\wspd = \{(A_1, B_1), \ldots, (A_s, B_s)\}$ with $A_i, B_i \subseteq S$ where:
\begin{itemize}
\item For every $p, q \in S$,
	there is a unique pair $(A_i, B_i) \in \wspd$ 
	such that $p \in A_i$ and $q \in B_i$ (or vice versa).

\item Let $c_{\min}(A_i, B_i)$ be the minimum distance between a $p \in A_i$ and $q \in B_i$.
	Every pair $(A_i, B_i) \in \wspd$ is \emph{well-separated}: 
	For a parameter $\delta > 0$,
	$\displaystyle \max\left\{\diam(A_i), \diam(B_i)\right\} \leq \delta \cdot c_{\min}(u, v)$.
\end{itemize}
A simple WSPD construction using compressed quadtrees is described 
in Chapter 3 of~\cite{har2011geometric}.
It produces a WSPD of size $O(n/\delta^d)$ in $O(n/\delta^d)$ time, provided a compressed quadtree on $S$.

We set $\delta = \tfrac{\eps}{2}$.
For a pair of quadtree nodes $u, v \in \uptree \times \qtree$, we define
$c_{\min}(u, v) := \min\{\dist{x}{y} \mid x \in \qtsquare_u, y \in \qtsquare_v \}$.
Using the algorithm in \cite{har2011geometric}, we compute a pair decomposition 
$\wspd \subseteq \uptreeV \times \downtreeV$ with the following properties:
\begin{enumerate}[label=(W\arabic*)]
\item \label{item:wspd_prop1}
	For every $(r, b) \in \reds \times \blues$,
	there is a unique pair $(u, v) \in \wspd$ such that 
	$r \in \reds_u, b \in \blues_v$.

\item \label{item:wspd_prop2}
	For every $(u,v) \in \wspd$,
	$\displaystyle \max\left\{\diam(\qtsquare_u), \diam(\qtsquare_v)\right\} \leq \delta \cdot c_{\min}(u, v)$.
	
\item \label{item:wspd_prop3}
	$|\wspd| = O(n/\eps^d)$.
\end{enumerate}
Given $T$, this algorithm constructs $\wspd$ in $O(n/\eps^d)$ time.
We set $\crossE = \wspd$ with each arc oriented from $\uptree$ to $\downtree$.
The cost of each arc in $\sparseG$ is set to:
\begin{equation*}
	c(u, v) = \begin{cases}
			\diam(\qtsquare_u \cup \qtsquare_v) & \text{if $(u, v) \in \crossE$} \\
			0 & \text{otherwise}
		\end{cases}
\end{equation*}
By triangle inequality and \ref{item:wspd_prop2}, this ensures that
\begin{equation}
\label{equation:wspd_cross_cost}
	\dist{r}{b} \leq c(u, v) \leq (1+\eps) \dist{r}{b}
\end{equation}
for any $r \in \reds$, $b \in \blues$, with $(u, v) \in \wspd$ the unique pair separating them.

To complete the description of the min-cost flow instance,
we specify the flow-supply/flow-demand $\fsupply'$ at vertices of $\sparseG$.
\begin{equation*}
	\fsupply'(v) = \begin{cases}
		\tsupply(\ell_\reds(v)) & \text{if $v$ is a red leaf in $\uptree$} \\
		-\tsupply(\ell_\blues(v)) & \text{if $v$ is a blue leaf in $\downtree$} \\
		0 & \text{otherwise} \\
	\end{cases}
\end{equation*}
Let $(\sparseG, c, \fsupply')$ be the resulting minimum cost flow instance.
The total time spent in constructing $\sparseG$ is $O(n\log n + n/\eps^d)$.

\subsection{Cost analysis}

Flow moves up from the leaves of $\uptree$,
through the cross edges into $\downtree$,
and finally descends to the sinks at leaves of $\downtree$.
By construction and \ref{item:wspd_prop1},
any pair $(r, b) \in \reds \times \blues$ has a unique path 
$p(r, b)$ from $r$ to $b$ in $\sparseG$ which uses a single cross edge.
We can map any transport $\transp$ (injectively) to a feasible pseudoflow 
$\flow_\transp$ on $G$, by placing a pseudoflow of $\transp(r, b)$ on $p(r, b)$.
Similarly, any pseudoflow $\flow$ can be mapped to a feasible 
transportation $\transp_\flow$ by \emph{decomposing} 
$\flow$ into pseudoflow on paths from supply vertices to demand vertices:
by the classical flow decomposition theorem,
$\flow$ can be decomposed 
into a set $\{\flow(p(r, b)\}$ of pseudoflows on the paths $p(r, b)$
(since $G$ is a directed acyclic graph,
any decomposition has no flow cycles, only paths).
Then, setting $\transp_\flow(r, b) = \flow(p(r, b))$ for all
$(r, b) \in \reds \times \blues$ is a feasible transportation.

Applying Equation~\ref{equation:wspd_cross_cost} to each path flow,
$\tfcost(\flow_\transp) \leq \tfcost(\transp)$ and
$\tfcost(\transp_\flow) \leq (1 + \eps) \tfcost(\flow)$.
We can apply these transformations to bound the approximation quality
of a transportation recovered by decomposing the \emph{optimal} pseudoflow 
$\flow^*$ of $\sparseG$.
\begin{equation*}
	\tfcost(\transp^*) \leq \tfcost(\transp_{\flow^*}) 
		\leq (1+\eps) \tfcost(\flow^*) 
		\leq (1+\eps) \tfcost(\flow_{\transp^*}) 
		\leq (1+\eps) \tfcost(\transp^*)
\end{equation*}
Then, by computing the optimal pseudoflow $\flow^*$ on $\sparseG$, 
a transportation solution $\transp_{\flow^*}$ constructed by a flow decomposition of 
$\flow^*$ meets the claimed approximation quality.
In the next subsection, we explain how to compute such a decomposition efficiently.
This cost analysis applies regardless of the specific flow decomposition of $\flow^*$.
Our recovery procedure in \ref{item:wspd_4} is a greedy decomposition.

\subsection{Recovering a transportation map}

Let $\flow^*$ be the optimal pseudoflow from $\reds$ to $\blues$ in $\sparseG$.
We use a two-part greedy algorithm to decompose $\flow^*$:
assigning the flow from $\reds$ to the cross edges through the up-tree,
then claiming the assigned flow using the down-tree.
Both steps amount to performing a flow decomposition on 
the portion of the flow lying in each tree,
treating the cross edge endpoints as sinks (resp., sources)
with demand equal to the sum of outgoing (resp., incoming) flow.
Both are arborescences, so flow decomposition can be done with a 
postorder traversal.
After solving both trees, we combine the paths assigned into
and out of each cross edge to find end-to-end path pseudoflows.

We only describe the decomposition for $\uptree$ in detail;
it is nearly identical for $\downtree$.
For each cross edge $(u, v) \in \crossE$, this produces lists 
$A_\reds(u, v)$ and $A_\blues(u, v)$ which hold pairs $(p, F)$
indicating {\em point $p \in \reds \cup \blues$ 
contributes $F$ units of the flow through $(u, v)$}.
The tree is processed in postorder: 
a node is visited only after all its children.
Denote the children of node $u$ by $\children(u)$.
Each node $u \in \uptree$ maintains a list $L(u)$ of the 
positive-demand red points in its subtree,
and a list $N(u)$ of the positive-flow cross edges leaving $u$.
$L(u)$ is initialized by joining lists $L(w)$ from each $w \in \children(u)$
(at the leaves, we initialize $L(r) = \{r\}$ for $r \in \reds$).
While $N(u)$ is not empty, let $(u, v) \in N(u)$ with flow $\flow^*(u, v) > 0$.
Take any point $r \in L(u)$ and add to $A_\reds(u, v)$ a pair
$(r, F)$, with $F = \min\{\fsupply(r), \flow^*(u, v)\}$,
also updating $\fsupply(r) \gets \fsupply(r) - F$ and 
$\flow^*(u, v) \gets \flow^*(u, v) - F$.
This has the effect of removing either $r$ from $L(u)$,
$(u, v)$ from $N(u)$, or both.
Once $N(u) = \emptyset$, all cross edges leaving $u$ have 
their flow assigned.

Finally, we complete the decomposition using $A_\reds(u, v)$ 
and $A_\blues(u, v)$, for each cross edge $(u, v)$.
While both $A_\reds(u, v)$ and $A_\blues(u, v)$ are nonempty,
let $(r, F_r) \in A_\reds(u, v)$ and $(b, F_b) \in A_\blues(u, v)$.
Output $\transp(r, b) := \min\{F_r, F_b\}$, update
$F_r \gets F_r - \transp(r, b)$ and $F_b \gets F_b - \transp(r, b)$,
and remove from the lists any pair $(p, F)$ for which $F = 0$.

We describe a charging scheme to prove this recovery routine
takes $O(n/\eps^d)$ time.
We charge the list union which constructs $L(u)$ to the children of $u$.
Each iteration performs a constant number of list operations and either
removes a node from $L(u)$, or a cross edge from $N(u)$.
Each removal occurs exactly once for every $r \in \reds$ 
and $(u, v) \in \crossE$,
so we charge iterations to the $r \in \reds$ or $(u, v) \in \crossE$
removed that iteration.
The processing of $A_\reds(u, v)$ and $A_\blues(u, v)$ can also
be charged per iteration to the pair $(p, F)$ removed in that iteration,
which is then charged back to the $p \in \reds \cup \blues$ or 
$(u, v) \in \crossE$ whose removal introduced $(p, F)$.
Thus, the total recovery time is
$O(|\sparseV| + |\crossE|) = O(n/\eps^d)$.

We computed an optimal pseudoflow $\flow^*: \sparseE \to \nats$ using 
the algorithm by Lee and Sidford~\cite{LS-flow}.
Since $|\sparseE| = O(n/\eps^d)$, the Lee-Sidford algorithm 
takes $\tilde{O}(n^{3/2} \eps^{-d} \polylog U)$ time;
where $U = \max_{p \in \reds \cup \blues} \fsupply(p)$ is the maximum demand,
solving the minimum cost flow problem dominates the running time 
over the construction of $\sparseG$ and recovery steps.
\begin{theorem}
\label{theorem:wspd_main}
	Let $\inst$ be an instance of the transportation problem in $\reals^d$ where $d$ is a constant.
	Let $\inst$ have size $n$,
	and let $\eps > 0$ be a constant.
	A transportation map $\transp$ for $\inst$ can be computed in
	$O(n^{3/2}\eps^{-d}\polylog n \polylog U)$ time with cost 
	$\tfcost(\transp) \leq (1 + \eps) \tfcost(\transp^*)$.
\end{theorem}

\subsection{Extension to doubling metric}

As in Section~\ref{section:grids}, 
we use an analagous structure for the grid-based partitioning.
Here, we seek to replace the WSPD built from a compressed quadtree.

Since we use the WSPD in a relatively black-box fashion,
it is sufficient to use the near-linear time WSPD construction by
Har-Peled and Mendel \cite{har2005fast}. 
This construction is a randomized algorithm which builds a WSPD for doubling metrics 
with $O(n\eps^{-O(D)})$ pairs in $O(2^{O(D)}n\log n + n\eps^{-O(D)})$ expected time.
To elaborate, their main result is a fast construction for an 
analogue to compressed quadtrees in doubling metrics called a \emph{net tree};
one can think of this as a linear-size version of the net hierarchy we discussed in Section~\ref{section:grids}.
The size of the net tree has no dependency on the spread.
Given the net tree, a WSPD is built using same greedy algorithm used for the compressed quadtree.
For our algorithm, we use the net tree whenever we would have used the compressed quadtree in the
Euclidean case.

\section{An Exact Algorithm}
\label{section:orlin}
In this section, we present an $\softO(n^2)$ time exact algorithm for the 
transportation problem in the plane.
To this end, we make modifications to an uncapacitated minimum cost flow 
algorithm due to Orlin~\cite{orlin1993faster}.
We begin by expanding the description of minimum cost flow from 
Section~\ref{section:wspd} to describe least-cost augmention algorithms.
The changes are documented in Section~\ref{subsection:orlin_algo},
and the resulting algorithm gives:

\begin{theorem}
\label{theorem:orlin_time}
	Given an instance $\inst = \tinst{}{}$
	of the transportation problem in $\reals^2$,
	an optimal transportation map can be computed in 
	$O(n^2 \polylog n)$ time.
\end{theorem}

Before we begin, we note that the algorithm requires strong connectivity of the 
input graph in order to define some distances. 
The transportation-to-MCF reduction we described in 
Section~\ref{subsection:mcf_def} is not strongly connected, but We can achieve 
strong connectivity by adding uncapacitated dummy arcs 
$\{(b, r) \mid (r, b) \in \reds \times \blues\}$ with sufficiently high cost, 
say $M \geq 2nUC$.
These dummy arcs have zero flow in an optimal flow.

\subsection{Preliminaries: improving flows and optimality conditions}

For the rest of this section, we work with an uncapacitated input network, 
where $u(v, w) = \infty$ for every $(v, w) \in E$.
Given a pseudoflow $f$, we define the capacitated \emph{residual network} as 
follows. 
For each arc $e = (v, w) \in E$, create an arc in the reverse direction 
$e^R = (w, v)$ and call the set of reverse arcs $E^R$.
We extend the costs to reverse arcs with $c(e^R) := -c(e)$ for $e^R \in E^R$.
Let the \emph{residual capacity} of each arcs in $E \cup E^R$ be
\begin{equation*}
	u_f(v, w) := \begin{cases}
		\infty & (v, w) \in E \\
		f(w, v) & (v, w) \in E^R
	\end{cases}
\end{equation*}
Let the \emph{residual arcs} be 
$E_f := \{(v, w) \in E \cup E^R \mid u_f(v, w) > 0\}$,
and the \emph{residual graph} be $G_f = (V, E_f)$.
The residual network is $N_f = (G_f, c, u_f, \fsupply)$.
Pseudoflows $f'$ in $G_f$ can \emph{augment} $f$ to produce a new pseudoflow:
the arc-wise addition $f + f'$ results in a valid pseudoflow for $G$.
If $f'$ is a residual pseudoflow with (i) $e_{f'}(v) = 0$ on vertices with 
$e_f(v) = 0$, (ii) $e_{f'}(v) \leq 0$ on vertices with $e_f(v) > 0$, and (iii) 
$e_{f'}(v) \geq 0$ on vertices with $e_f(v) < 0$, and 
(iv) $\sum_{v \in V}|e_{f'}(v)| < \sum_{v \in V}|e_f(v)|$, 
we call it an \emph{improving flow}.
Intuitively, an improving flow routes some amount of excess to deficits,
reducing the total amount of excess/deficit.
If an improving flow $f'$ is a flow on a simple path from an excess vertex to a 
deficit vertex, we call it an \emph{augmenting path}.
Generally, an improving flow $f'$ may be any residual pseudoflow (e.g. a 
blocking flow), but the improving flows of this algorithm will be augmenting 
paths.

Next, we introduce some criteria for proving a flow is mininum-cost.
The MCF problem can be expressed as a linear program with variables $f$.
\begin{equation*}
\begin{aligned}
	& \min & & \sum_{(v, w) \in E}{c(v, w) f(v, w)} & \\
	& \text{s.t.} & & \sum_{(v, w) \in E} f(v, w) - \sum_{(w, v) \in E} f(w, v) = \fsupply(v) & \forall v \in V \\
	& & & f(v, w) \geq 0 & \forall (v, w) \in E 
\end{aligned}
\end{equation*}
The dual program has a variable for each vertex, that we call 
\emph{potentials} $\dual : V \to \reals$.
\begin{equation*}
\begin{aligned}
	& \max & & \sum_{v \in V}{\fsupply(v) \dual(v)} & \\
	& \text{s.t.} & & \dual(v) - \dual(w) \leq c(v, w) & \forall (v, w) \in E
\end{aligned}
\end{equation*}
The dual feasibility conditions are often restated in terms of the 
\emph{reduced cost}
\begin{equation*}
	c_\dual(v, w) := c(v, w) - \dual(v) + \dual(w).
\end{equation*}
That is, a feasible $\dual$ has $c_\dual(e) \geq 0$ for all arcs.
The complementary slackness conditions of this LP can also be written using 
reduced cost.
\begin{equation}
\label{eqn:optimality}
\begin{aligned}
	c_\dual(v, w) > 0 &\implies f(v, w) = 0 \\
	f(v, w) > 0 &\implies c_\dual(v, w) = 0
\end{aligned}
\end{equation}
The complementary slackness conditions are \emph{optimality conditions} --- 
$f^*$ is the minimum-cost flow if and only if there exists potentials $\dual^*$
that satisfy (\ref{eqn:optimality}).
We say that pseudoflow $f$ and potentials $\dual$ ``satisfy the optimality
conditions'' if together they meet (\ref{eqn:optimality}).
An equivalent statement of (\ref{eqn:optimality}), in terms of the residual 
graph, is that $c_{\dual^*}(v, w) \geq 0$ on all arcs of $G_{f^*}$.

\subsection{Least-cost augmentation algorithms}

A number of primal-dual algorithms for min-cost flow (including 
\cite{orlin1993faster}) augment flow along shortest paths in $G_f$ with respect 
to reduced costs.
We will refer to such algorithms as \emph{least-cost augmentation} algorithms.

Least-cost augmentation algorithms typically start with $f=0, \dual=0$, a 
pseudoflow-potential pair which trivially meets the optimality conditions, 
but where $f$ will generally have excess and not be a flow.
The excess/deficit is gradually reduced to 0 by augmenting with improving flows
along shortest paths in $G_f$ w.r.t. $c_\dual$.
Using improving flows along shortest paths guarantees that the new pseudoflow 
meets the optimality conditions with some set of feasible potentials.

\begin{lemma}[\cite{orlin1993faster} Lemma 4]
\label{lemma:leastcost}
	Let $f$ be a pseudoflow satisfying the optimality condition with 
	feasible potentials $\dual$, and $f'$ be derived from $f$ by augmenting
	along a shortest path in $G_f$ with respect to costs $c_\dual$.
	Then $f'$ also satisfies the optimality conditions with respect to
	another set of feasible potentials $\dual'$.
\end{lemma}

Thus, when $f$ has 0 imbalance everywhere, it must be an optimal flow.
The proof of Lemma~\ref{lemma:leastcost} includes a construction for the new 
potentials $\dual'$ in linear time, after computing the single-source
shortest path (SSSP) distances $d_s(\cdot)$ from one excess vertex $s$.
Specifically, $\dual'(v) = \dual(v) - d_s(v)$.

\subsection{Excess Scaling and Orlin's Algorithm}
\label{subsection:orlin_algo}

Least-cost augmentation algorithms (with augmenting paths) can be tuned by 
choosing which excess/deficit vertices to augment between, and the amount of 
flow sent per augmentation.
Orlin's algorithm is based on the \emph{excess-scaling} paradigm due to Edmonds 
and Karp~\cite{edmonds1972theoretical}, which augments flow in units of 
$\Delta$ (initially $\Delta = U$) from excess vertices with 
$e_f(v) \geq \Delta$ to deficit vertices with $e_f(v) \leq -\Delta$. 
Excess/deficit vertices which meet these thresholds are called \emph{active}.
Augmentations are performed from any active excess vertex to any active 
deficit vertex.
When there are either no more active excesses or no active deficits, $\Delta$ 
is halved.
Each sequence of augmentations with a fixed $\Delta$ is called a 
\emph{scaling phase}, and $\Delta$ is called the \emph{scale}.
As an invariant, there is at most $2n\Delta$ excess at the beginning of each scale,
and therefore $O(n)$ augmentating paths in a scaling phase before the active
excess vertices or active deficit vertices are depleted and the phase ends.
There are $O(\log U)$ scales before $\Delta < 1$. 
At this point, the remaining (integer) excess must be 0.

The algorithm in \cite{orlin1993faster} extends excess-scaling 
to obtain a strongly polynomial $O(n\log n)$ bound on the number of 
augmentations and scaling phases. 
At a high level, this is done by contracting edges, adjusting the set of active
vertices, and reducing $\Delta$ aggressively when possible.

We use $\hat{G} = (\hat{V}, \hat{E})$ to refer to the contracted multigraph, 
where each $\hat{v} \in \hat{V}$ is a contracted subset of $V$ and 
$\hat{E} = \{(v, w) \in E \mid v \in \hat{v}, w \in \hat{w}, \hat{v} \neq \hat{w}, \hat{v} \in \hat{V}, \hat{w} \in \hat{V}\}$.
Initially, $\hat{G}$ is identical to $G$ and each $\hat{v}$ is a singleton.
We refer to the elements of $\hat{V}$ as \emph{supervertices}.
The imbalance of $\hat{v} \in \hat{V}$ is defined to be 
$e_f(\hat{v}) = \sum_{v \in \hat{v}} e_f(v)$.
The algorithm maintains potentials for every $v \in V$, but only flow values
for the edges of $\hat{E}$.
Namely, flow values across of arcs within $\hat{v} \in \hat{V}$ are unknown.

The algorithm executes a single scaling phase as follows:
\begin{enumerate}
	\item Using any $1/2 < \alpha < 1$, excess supervertices (resp. deficit 
		supervertices) are active if $e_f(v) \geq \alpha\Delta$ (resp. 
		$e_f(v) \leq -\alpha\Delta$).

	\item If there are no active supervertices and zero flow on all 
		$\hat{E}$ arcs, $\Delta$ is set to 
		$\max_{\hat{v} \in \hat{V}} e_f(\hat{v})$.

	\item For any $(v, w) \in \hat{E}$ has $f(v, w) \geq 3n\Delta$, it is 
		contracted, merging $\hat{u} \gets \hat{v} \cup \hat{w}$,
		where $v \in \hat{v}$ and $w \in \hat{w}$.
		Set the flow of all arcs in $\hat{v} \times \hat{w}$ beside
		$(v, w)$ to 0.

	\item Otherwise, repeatedly perform least-cost augmentation from any
		active excess supervertex to any active deficit supervertex,
		in units of $\Delta$, until there are no more active excess
		supervertices or no more active deficit supervertices.
		Each augmentation involves the following steps:
		\begin{enumerate}
		\item Compute the SSSP tree from any $s \in \hat{v}$ of an 
			active excess supervertex $\hat{v}$.
			Let the SSSP distances from $s$ be $d_s(\cdot)$.
		\item Update potentials for all $v \in V$ as
			$\dual(v) \gets \dual(v) - d_s(v)$.
		\item Augment $\Delta$ flow from $s$ to any $t \in \hat{w}$ of 
			an active deficit supervertex $\hat{w}$, along the 
			shortest path from $s$ to $t$.
		\end{enumerate}
		
	\item $\Delta \gets \Delta/2$
\end{enumerate}
Intuitively, arcs are contracted when they have flow so high that no set of
future augmentations can bring the flow back to 0 --- these arcs are members
of the optimal flow's support.
As a result, the vertices within a supervertex are connected by a set of 
positive flow edges, although the exact flow value is unknown.
In the uncapacitated setting, $f(v, w) > 0$ induces $u_f(v, w) > 0$ and 
$u_f(w, v) > 0$, so the members of a supervertex are strongly connected by 
positive capacity residual edges. %TODO move the scc thing into a property?

%TODO might need a lmema for this or something
During a contraction of $\hat{v}$ and $\hat{w}$ ((3) above), we also set flow 
to 0 on all arcs between $\hat{v} \times \hat{w}$ besides the contracted arc.
Note that this does not change the imbalance of the new supervertex,
but the only flow support arcs within supervertices are contracted 
arcs.
As a result, our augmenting paths will only use contracted
arcs to traverse the interior of supervertices.
Thus, augmenting paths which pass through the interior of supervertices
will not change the flow support status of arcs internal to the supervertex.
We can track the flow support within supervertices by marking the 
contracted arcs as having positive flow and ignore the exact value, and this
status will not change throughout the augmentation process.

The output flow is one on the contracted graph, not $G$.
In contrast, the resulting optimal potentials are on $V$.
An extra step is taken to recover an optimal flow on $G$ from the optimal 
potentials, which we describe in Section~\ref{subsection:orlin_recovery}.
Our solution to this is new and extracts the optimal flow in near linear time.
Excepting this recovery step, Orlin bounds the running time of the other steps
in terms of the time to solve the SSSP problem.

\begin{theorem}[\cite{orlin1993faster} Theorem 4]
\label{theorem:orlin_orlin}
	The algorithm determines the minimum cost flow in the contracted 
	network in $O((n\log n) S(n, m))$ time, where $S(n, m)$ is the time
	to compute the single-source shortest path tree/distances in a graph of 
	$n$ vertices and $m$ edges with nonnegative costs.
\end{theorem}

In Section~\ref{subsection:orlin_spt}, we provide a geometric implementation of
Dijkstra's algorithm by using a dynamic data structure for bichromatic closest 
pair.
This approach is not new to geometric transportation/matching algorithms, 
see \cite{vaidya1989geometry,atkinson1995using,sharathkumar2012algorithms}.
In the end, $S(n, m) = \softO(n)$, completing the proof of 
Theorem~\ref{theorem:orlin_time}.
From this implementation, we can also prove the flow support has size $O(n)$ 
throughout the algorithm, meaning steps which check properties of flow support 
edges (e.g. (2) and (3) above) take only $O(n)$ time.

\subsection{Computing shortest path trees}
\label{subsection:orlin_spt}

Consider Dijkstra's algorithm for computing shortest distances $d_s(v)$ to all 
vertices from a single source $s \in V$.
Dijkstra's algorithm expands a set $S \subseteq V$ of vertices for which it has 
calculated $d_s$.
Initially $S = \{s\}$, $d_s(s) = 0$, and $d_s(v)$ is unknown for all $v \neq s$.
In each iteration, Dijkstra's algorithm grows $S$ by \emph{relaxing} 
the edge to the minimum-distance neighbor of $S$
\begin{equation*}
	(v, w) = \argmin_{(v, w) \in S \times (V \setminus S)} d_s(v) + c_\dual(v, w),
\end{equation*}
fixing $d_s(w) \gets d_s(v) + c_\dual(v, w)$ and adding $w$ to $S$.
The shortest path tree from $s$ is composed of the relaxed edges
and can therefore be computed alongside the distances.

Geometrically, we can find this $\argmin$ efficiently, i.e. without examining 
$\Theta(n^2)$ edges over the course of Dijkstra's.
Our solution is to use a dynamic bichromatic closest pair data structure.
Let $P, Q \subset \reals^2$ be point sets with weights 
$W: P \cup Q \to \reals$.
The \emph{bichromatic closest pair} (BCP) between $P$ and $Q$ is the pair 
$(p, q) \in P \times Q$ minimizing the additively weighted distance 
$\dist{p}{q} + \omega(p) + \omega(q)$.
The dynamic BCP data structure of Kaplan~\etal~\cite{kaplan2017dynamic}
maintains the BCP under insertion and deletion of points and implements each 
update/query operation in $\softO(1)$ time.

\begin{lemma}
	Using a dynamic bichromatic closest pair data structure,
	we can compute the shortest path tree and distances under reduced costs
	in $\softO(n)$ time.
\end{lemma}

We emphasize that the BCP query requires points in $\reals^2$, so will act on 
the original points in $V$ rather than supervertices directly.
To choose the starting vertex $s$, we use any $s = v \in V$ such that $v$ is a
member of an active excess supervertex $\hat{v} \in \hat{V}$.
The vertices within a supervertex are strongly connected in $G_f$ by flow 
support arcs, which must have reduced cost 0, by the optimality conditions. 
Thus, all members of a supervertex have the same $d_s$ distance.

\subsubsection*{Dijkstra's algorithm implementation}

For the BCP, we maintain $P = (S \cap \reds)$, 
$Q = ((V \setminus S) \cap \blues)$, and use weights 
$\omega(v) = d_s(v) + \dual(v)$ 
for $v \in P$ and $\omega(w) = -\dual(w)$ for $w \in Q$.
Consider an iteration of Dijkstra's algorithm, and let 
$v \in S, w \in V \setminus S$.
If $(v, w) \in E$ (i.e., directed from $\reds$ to $\blues$), then 
$d_s(v) + c_\dual(v,w) = \dist{v}{w} + \omega(v) + \omega(w)$ and the BCP on 
$P$ and $Q$ above will accurately find the $\argmin$ for these edges.

However, this BCP will not correctly report the minimum edge if it it is 
$(v, w) \in E_f \setminus E$ (directed from $\blues \to \reds$).
These are residual edges which ``undo'' existing flow support.
On the other hand, when $(v, w)$ is in the flow support, we must have 
$d_s(v) = d_s(w)$.
Dijkstra's algorithm maintains that $d_s(v) \leq d_s(w)$ for 
$v \in S, w \in (V \setminus S)$, so such a support edge is always a minimizer.
We prioritize existing support edges by relaxing any new support edges leaving 
$S$ before querying the BCP for $\reds \to \blues$ edges.
This can be done in total time proportional to the number of support edges, by 
storing the support edges adjoining each vertex of $\blues$ in a list.

%TODO this might need a figure
Finally, we note that the search may become ``stuck'' and unable to reach 
certain vertices of $\reds$ if they are inactive and have no adjoining flow 
support.
We call these vertices \emph{dead}.
For these, we relax the high-cost dummy edges mentioned in the introduction.
Since the dummy edges have uniform cost, the minimum dummy edge to reach some 
$r \in \reds$ is exactly $(b, r)$ for 
$b = \argmax_{b \in S \cap \blues} \dual(b)$.
Since this is independent of the arc cost, the same $b$ is the minimizer for 
reaching every dead vertex.
Thus, we can compute $d_s$ for all dead vertices in a single step at the end
of Dijkstra's, by relaxing the dummy edges leaving the $b$ vertex above.

Each relaxation, we add $w$ to $S$ and update $d_s(w)$, $P$, $Q$, and 
$\omega(w)$.
As before, we stop once $S = V$.
In summary, our implementation of Dijkstra's algorithm is as follows.
\begin{enumerate}
	\item Initialize $S = \{s\}$ and $d_s(s) = 0$, and $d_s(v) = \infty$ 
		for all other vertices. 
		Initialize an empty queue reached support edges.

	\item Try to relax the next support edge $(v, w)$ in the queue.
		If $S$ has grown so that $w$ is no longer in $V \setminus S$,
		then discard the edge and try again.

	\item Otherwise, if there are no support edges, query the BCP to relax 
		the minimum $\reds \to \blues$ edge.
		Repeat from step $2$ until $S = V$ or 
		$(V \setminus S) \cap \blues = \emptyset$.

	\item If there are any dead vertices $r \in \reds$, relax every 
		$(b, r)$ dummy edge for 
		$b = \argmax_{b \in S \cap \blues} \dual(b)$.
\end{enumerate}

We now analyze the time spent in computing shortest paths.
Because we relax support edges before using the BCP data structure, we have the
following lemma.

\begin{lemma}
\label{lemma:orlin_acyclic}
	The flow support is acylic, in the undirected sense, throughout our 
	implementation of Orlin's algorithm.
	That is, there is no sequence of edges 
	$(v_1, v_2), (v_2, v_3) \cdots (v_k, v_{k+1} = v_1)$
	for which $\forall 1 \leq i \leq k$ either $f(v_i, v_{i+1}) > 0$ or 
	$f(v_{i+1}, v_i) > 0$.
\end{lemma}
\begin{proof}
Let $f_i$ be the pseudoflow after the $i$-th augmentation, and 
$F_i \subseteq E$ be the flow support of $f_i$.
Let $T_i$ be the shortest path tree generated for augmenting $f_i$.
Namely, the $i$-th augmenting path is an excess-deficit path in $T_i$, and all
arcs of $T_i$ are admissible by the time the augmentation is performed.
Let $E(T_i)$ be the undirected edges corresponding to arcs of $T_i$.
Notice that, $E(F_{i+1}) \subseteq E(F_i) \cup E(T_i)$.
We prove that $E(F_i) \cup E(T_i)$ is acyclic by induction on $i$;
as $E(F_{i+1})$ is a subset of these edges, it must also be acyclic.
At the beginning with $f_0 = 0$, $E(F_0)$ is vacuously acyclic.

Let $E(F_i)$ be acyclic by induction hypothesis.
Since $T_i$ is a forest (thus, acyclic), any hypothetical cycle $\Gamma$ that
forms in $E(F_i) \cup E(T_i)$ must contain edges from both
$E(F_i)$ and $E(T_i)$.
To give a visual analogy, we will color $e \in \Gamma$
\emph{purple} if $e \in E(F_i) \cap E(T_i)$,
\emph{red} if $e \in E(F_i)$ but $e \not\in E(T_i)$,
and \emph{blue} if $e \in E(T_i)$ but $e \not\in E(F_i)$.
Then, $\Gamma$ is neither entirely red nor entirely blue.
We say that red and purple edges are \emph{red-tinted}, and similarly blue and
purple edges are \emph{blue-tinted}.
Roughly speaking, our implementation of the Hungarian search prioritizes
relaxing red-tinted admissible arcs over pure blue arcs. %TODO figure

We can sort the blue-tinted edges of $\Gamma$ by the order they were relaxed
into $S$ during the Hungarian search forming $T_i$.
Let $(v, w) \in \Gamma$ be the last pure blue edge relaxed, of all the
blue-tinted edges in $\Gamma$ --- after $(v, w)$ is relaxed, the remaining
unrelaxed, blue-tinted edges of $\Gamma$ are purple.

Let us pause the Hungarian search the moment before $(v, w)$ is relaxed.
At this point, $v \in S$ and $w \not\in S$, and the Hungarian search must have
finished relaxing all frontier support arcs.
By our choice of $(v, w)$, $\Gamma \setminus (v, w)$ is a path of relaxed blue
edges and red-tinted edges which connect $v$ and $w$.
Walking around $\Gamma \setminus (v, w)$ from $v$ to $w$, we see that every
vertex of the cycle must be in $S$ already: $v \in S$, relaxed blue edges have
both endpoints in $S$, and any unrelaxed red-tinted edge must have both
endpoints in $S$, since the Hungarian search would have prioritized relaxing
the red-tinted edges to grow $S$ before relaxing $(v, w)$ (a blue edge).
It follows that $w \in S$ already, a contradiction.

No such cycle $\Gamma$ can exist, thus $E(F_i) \cup E(T_i)$ is acyclic
and $E(F_{i+1}) \subseteq E(F_i) \cup E(T_i)$ is acyclic.
By induction, $E(F_i)$ is acyclic for all $i$.
\end{proof}

Lemma~\ref{lemma:orlin_acyclic} implies there are at most~$O(n)$ edges in the 
flow support at any point.
Each support edge enters/exits the Dijsktra search queue at most once, so the 
time spent attempting to relax support edges is $O(n)$.
There are $O(n)$ additional iterations where we query the BCP data structure,
and we spend $O(n)$ time relaxing dummy edges at the end.
The total time to compute each shortest paths tree is $\softO(n)$, so the main 
scaling/augmentation procedure can be implemented in $\softO(n^2)$ time.

\subsection{Optimal flow from optimal duals}
\label{subsection:orlin_recovery}

Let $T^*$ be the shortest path tree under $c_{\dual^*}$ rooted at an arbitrary 
$s \in \reds$, computed using the modified Dijkstra's algorithm in the previous 
section.
The following lemma completes the proof of the 
Theorem~\ref{theorem:orlin_time}.
\begin{lemma}
	$T^*$ does not use any dummy edges, and there exists a flow in $T^*$ 
	which satisifies all supplies and demands.
\end{lemma}
\begin{proof}
	Note that $T^*$ is constructed based on the final flow support in 
	addition to the optimal dual $\pi^*$.
	Recall that dummy edges cannot be flow support edges, and that dummy
	edges are essentially only relaxed for $\reds$ points without adjoining
	flow support.
	Assume that all vertices of $\reds$ have nonzero supply (indeed, 
	discarding zero-supply red vertices results in an equivalent instance).
	If $r \in \reds$ was not contracted, then it has an adjoining 
	uncontracted flow edge to offload its supply.
	Alternatively, if it is part of a contracted component, then there is a
	support edge which resulted in $r$'s contraction, although we may not
	know the flow value on the edge.

	For the second part, we will start from the empty flow and consider the 
	linear program description of min-cost flow.
	There are dual variables for each vertex and dual constraints for each 
	edge.
	First, we argue that $T^*$ forms a maximal set of tight, linearly 
	independent dual constraints.
	We observe that a set of edges $S \subset E$ correspond to a set of 
	linearly independent dual constraints exactly when $S$ is acyclic 
	(in the undirected sense).  Since $T^*$ spans all vertices, its edges 
	have linearly independent constraints and the set is maximal.

	Next, we show that this set of edges must contain a feasible primal 
	solution, by a perturbation argument.
	Suppose we raise the cost of all edges in $E \setminus T^*$ to $\infty$ 
	(or any quantity $\geq U \tfcost(\dual^*)$, for example).
	This does not change the optimal dual solution; $\dual^*$ is still 
	``stuck'' due to the maximal set of independent, tight constraints in 
	$T^*$.
	However, this forces all other dual constraints to be slack on 
	$\dual^*$.
	
	Strong duality implies that an equal-cost solution exists in the 
	primal, since $\tfcost(\dual^*)$ is finite.
	However, the perturbation has caused any solution which uses edges 
	outside of $T^*$ to be much more than $\tfcost(\dual^*)$,
	so the remaining optimal primal solution must be within the edges of 
	$T^*$.
\end{proof}

{\small
\bibliographystyle{abbrv}

}

\end{document}